\documentclass[12pt, reqno]{amsart}

\usepackage{amsmath}
\usepackage{amssymb} 
\usepackage{array}
\usepackage{enumitem}
\usepackage{color}
\definecolor{battleshipgrey}{rgb}{0.52, 0.52, 0.51} 
\usepackage{hyperref}
\hypersetup{
    bookmarks=true,         
    unicode=false,          
    pdftoolbar=true,        
    pdfmenubar=true,        
    pdffitwindow=false,     
    pdfstartview={FitH},    
    pdftitle={Essay on axioms of quantum mechanics. II},    
    pdfauthor={Wolfgang Bertram},     
    pdfsubject={Subject},   
    pdfcreator={Creator},   
    pdfproducer={Producer}, 
    pdfkeywords={keyword1} {key2} {key3}, 
    pdfnewwindow=true,      
    colorlinks=true,       
    linkcolor=battleshipgrey,          
    citecolor=battleshipgrey,        
    filecolor=battleshipgrey,      
    urlcolor=battleshipgrey          
}
\usepackage{verbatim} 

\usepackage{pstricks-add}
\usepackage{pst-func}

\theoremstyle{plain}
\newtheorem{theorem}{Theorem}[section]
\newtheorem{lemma}[theorem]{Lemma}
\newtheorem{proposition}[theorem]{Proposition}

\newtheorem{definition}[theorem]{Definition}
\theoremstyle{remark}
\newtheorem{remark}{Remark}[section]
\newtheorem{example}{Example}[section]

\numberwithin{equation}{section}

\newcommand{\bA}{\mathbb{A}}

\newcommand{\K}{\mathbb{K}}

\newcommand{\R}{\mathbb{R}}
\newcommand{\bC}{\mathbb{C}}

\newcommand{\Z}{\mathbb{Z}}

\newcommand{\PP}{\mathbb{P}}

\newcommand{\bV}{\mathbb{V}}

\newcommand{\cS}{{\mathcal S}}

\newcommand{\cH}{\mathcal{H}}

\newcommand{\cR}{\mathcal{R}}
\newcommand{\cO}{\mathcal{O}}

\newcommand{\g}{\mathfrak{g}}

\newcommand{\uu}{\mathfrak{u}}

\newcommand{\fS}{\mathfrak{S}}

\newcommand{\kk}{\sl{k}} 

\newcommand{\Gl}{\mathrm{Gl}}
\newcommand{\GL}{\mathrm{GL}}
\newcommand{\SL}{\mathrm{SL}}
\newcommand{\UU}{\mathrm{U}}

\newcommand{\SO}{\mathrm{SO}}

\newcommand{\Herm}{\mathrm{Herm}}

\newcommand{\Bij}{\mathrm{Bij}}

\newcommand{\id}{\mathrm{id}}
\newcommand{\Aherm}{\mathrm{Aherm}}

\newcommand{\Ad}{\mathrm{Ad}}

\newcommand{\Graph}{\mathrm{Graph}}

\newcommand{\CR}{\mathrm{CR}}
\newcommand{\tr}{\mathrm{trace}}

\newcommand{\ant}{\mathbf{\infty}}

\newcommand{\msk}{\medskip}
\newcommand{\ssk}{\smallskip}
\newcommand{\nin}{\noindent}

\begin{document}

 \title[An Essay on the completion of Quantum Theory. II]{An Essay on the completion of Quantum Theory.
 \\  II:  Unitary time evolution} 

% ... 

\author{Wolfgang Bertram}

\address{Institut \'{E}lie Cartan de Lorraine \\
Universit\'{e} de Lorraine at Nancy, CNRS, INRIA \\
B.P. 70239 \\
F-54506 Vand\oe{}uvre-l\`{e}s-Nancy Cedex, France}

\email{\url{wolfgang.bertram@univ-lorraine.fr}}

\begin{abstract}
In  this second part of the ``essay on the completion of quantum theory'' we define the {\em unitary setting of completed quantum mechanics},
by adding as intrinsic data  to those from Part I (\cite{Be17}) the choice of a {\em north pole $N$ and south pole $S$} in the geometric space $\cS$.
Then we 
explain that, in the unitary  setting, a  complete observable
corresponds to a
 right (or left) invariant vector field (Hamiltonian field)  on the geometric space, and {\em unitary time evolution} is the flow of such a vector field.
This interpretation is in fact nothing but the Lie group-Lie group algebra correspondence, for a geometric space that can be interpreted as the
{\em Cayley transform} of the usual, Hermitian operator space. In order to clarify the geometric nature of this setting, we realize the 
Cayley transform as a member of a natural {\em octahedral group} that can be associated to any triple of pairwise transversal elements.  
 \end{abstract}



\maketitle

\vskip 10mm

\begin{center}
{\sl
Again, dedicated to the memory of  \href{http://www.iecl.univ-lorraine.fr/~Wolfgang.Bertram/NachrufBrandes.pdf}{Tobias}.}
\end{center}

\bigskip

\section{Introduction}

\subsection{Mathematical core of axiomatics:   Jordan-Lie algebras}
{\em It is an important feature of quantum mechanics that the physical variables play a dual role, as observabes
{\em and} as generators of transformation groups....
}\footnote{Alfsen and Schultz, \cite{AS}, p.vii }
Actually, this important feature  shows up both in classical and in quantum mechanics: 
classically, a function $H$ plays the role of an observable,
and its associated Hamiltonian vector field $\xi_H$ is the generator of time evolution.
In quantum mechanics, a Hermitian operator $H$ plays the role of an observable, and its associated 
skew-Hermitian operator $i H$, or rather $X_H:=\frac{1}{ i \hbar} H$, is the generator of time evolution:  this is expressed by  the 
general form of the 
\href{https://en.wikipedia.org/wiki/Schrodinger_equation}{Schr\"odinger  equation}
\begin{equation}\label{eqn:Sch}
\partial_t \psi = - \frac{i}{\hbar} H \psi   .
\end{equation}
When $H$ is seen as observable, we consider it as element of the {\em Jordan algebra} $\Herm(\cH)$, 
and when we see it as generator  of  a
transformation group, we consider it as an element of the {\em Lie algebra} $i \Herm(\cH)$. 
To understand both aspects simultaneously, we must regard the space $\Herm(\cH)$ both as Lie and Jordan algebra: both
structures are compatible and define what one calls a {\em Jordan-Lie algebra (with positive Jordan-Lie constant)}.
Thus, for developing an axiomatic theory, and for understanding the mathematical core of what is going on, it
seems that Jordan-Lie algebras are the correct starting point -- this point of view has been advocated by Emch, see \cite{E}, and also \cite{L98}.
The first chapter of this text gives a self-contained introduction to Jordan-Lie algebras; some additional material  can be found in
an appendix (appendix \ref{app:JLA}).

\subsection{Physics: on ``pictures''.}
A striking feature of unitary time evolution in Quantum Mechanics is, as every student learns, that there are
two, or three, ``pictures'' describing it: the 
\href{https://en.wikipedia.org/wiki/Schrodinger_picture}{\em Schr\"odinger picture} (states evolve, observables are constant),
the \href{https://en.wikipedia.org/wiki/Heisenberg_picture}{\em Heisenberg picture} (states are constant, observables evolve), and the 
\href{https://en.wikipedia.org/wiki/Interaction_picture}{\em interaction picture}
(both evolve). Naively, one may ask: which one is the ``correct'' one, or the ``most realistic'' one?
As one learns, all are ``correct'':   they are mathematically equivalent.
Indeed, already on
the level of classical mechanics, there are two such pictures, the ``Hamilton picture'' and the ``Liouville picture''
(cf.\  \cite{Tak08}, p. 59-60).
The Liouville-Schr\"odinger picture is the perception, say,  of a newspaper reader,
 that the ``state of the world'' evolves, and we sit in our armchair and ``observe it evolving''.
The Hamilton-Heisenberg picture rather describes the world as a vast landscape, seen by an observer sitting
 in a fast train and watching the
landscape going by the window: our perspective, aka observable, changes every second, but we know that the landscape outside is
stable and immobile... Both viewpoints are right -- change is real, but who moves? Duality is real, too: we need to have (at least) two categories of 
stuff, if we  want to speak of one changing with respect to another.

\subsection{Yet another picture}
In the framework of ``completed quantum theory'', I will give another picture of time evolution,
again mathematically equivalent to the Schr\"odinger or Heisenberg picture.
The basic idea is simple, and familiar to all mathematicians and physicists having some working knowledge in 
{\em Lie groups}. 
Namely, recall that time evolution based on the Schr\"odinger equation (\ref{eqn:Sch}), applied to mixed states $W$ (density matrices)
takes the form
\begin{equation}\label{eqn:Sch2}
t \mapsto W_t=   
 e^{- \frac{it}{\hbar} H} W e^{ \frac{it}{\hbar} H} ,
\end{equation}
where $e^{ \frac{it}{\hbar} H}$ is a {\em unitary} operator. Moreover, up to replacing $t$ by $-t$, 
\href{https://en.wikipedia.org/wiki/Heisenberg_picture#Summary_comparison_of_evolution_in_all_pictures}{\em
this is the same form as time evolution of observables in the Heisenberg picture}.
This fact stresses again that observables and density matrices behave dually to each other, 
and it describes precisely the {\em action of the unitary group $\UU(\cH)$ by conjugation on its Lie algebra $i\Herm(\cH)$},
respectively on its dual.
In mathematical language, {\em usual time evolution is the adjoint, respectively the coadjoint, representation of the (infinite dimensional)
Lie group $\UU=\UU(\cH)$ on its Lie algebra, resp.\ on its dual space}.
(Well, there still is an additional factor $\hbar$ about which we have to speak later...) 

\ssk
From this viewpoint, 
the next step will appear natural:  the adjoint representation comes from a global and geometric action of the 
Lie group $\UU$,  namely its {\em left and right action  on itself}.  
In the spirit of ``delinearization of quantum mechanics'' (cf.\ Part I), it seems thus natural 
to interprete {\em time evolution as a flow coming from left (or right) action of the unitary group on itself}.  
Just like for any Lie group,  the adjoint and coadjoint action  are certain means
to describe the left and right action by {\em trivializing the tangent bundle $T\UU$ of $\UU$}, hence are certain
``pictures'' of a more geometric setting. 
The reason why this geometric setting is never emplyoed in ``usual'' quantum mechanics is simply that usual quantum mechanics 
requires observables or density matrices to live in a {\em linear} space, like a Lie or a Jordan algebra, whereas the group
$\UU$ itself is a ``geometric'' and  ``non-linear'' object. But this is precisely the point where the setting of completed quantum theory
comes in: the upshot is  that the completed space $\cR$ of our theory (Part I) can be identified with $\UU$;
the real universe $\cR$ ``is a group'' (albeit without origin singled out: a ``group without fixed origin'', sometimes also called
a {\em torsor}). 
This feature, too, is in principle well-known: namely, via the
\href{https://en.wikipedia.org/wiki/Cayley_transform#Operator_map}{\em Cayley transform}, the set of Hermitian operators can
be identified with an (open) subset of the unitary group $\UU = \UU(\cH)$. 
The important question is therefore: {\em in what sense can we consider this identification as ``canonical''?
What does the identification imply for the shape of the theory?}
Our answer to this question can be summarized as follows:
\begin{itemize}
\item
in completed, geometric quantum theory, the unitary group structure on $\cR$ is canonical, {\em provided the pair $(N,S)$, where $N$ is the
 ``north pole''  and $S$ the ``south pole'', is considered as fixed datum of our theory} (both poles lie in the ``complex universe'' $\cS$, but {\em outside}
  the real form $\cR$),
\item
in ``ordinary'', linear quantum theory, the unitary group structure on $\cR$
 can only be seen indirectly, because the points $0$ and $\infty$ are not fixed under the geometric action:
  one needs some transformation relating the various ``pictures'', the most important being the Cayley transform.
The Cayley transform itself is not ``canonical''; it's just  a tool -- but a tool whose ``geometry'' we have to analyze.
\end{itemize}

\subsection{Image of the picture}
To fix ideas, and to provide an image, let us consider the case of the (commutative) one-dimensional $C^*$-algebra
$\bA = \bC$, whose completion is the {\em Riemann sphere} $\cS=\bC\PP^1$.
\begin{figure}[h]
\caption{The Riemann sphere $\cS$ with six poles
 $(N,S,O,\infty,F,B)$.}\label{fig:RS}
\newrgbcolor{eqeqeq}{0.8784313725490196 0.8784313725490196 0.8784313725490196}
\newrgbcolor{aqaqaq}{0.6274509803921569 0.6274509803921569 0.6274509803921569}
\psset{xunit=0.6cm,yunit=0.6cm,algebraic=true,dimen=middle,dotstyle=o,dotsize=5pt 0,linewidth=0.8pt,arrowsize=3pt 2,arrowinset=0.25}
\begin{pspicture*}(-9.68,-4.78)(13.66,4.94)
\psplotImp[linewidth=1.8pt](-11.0,-8.0)(14.0,7.0){-1.0+1.0*y^2+0.0625*x^2}
\psplotImp(-11.0,-8.0)(14.0,7.0){-1.0+0.0625*y^2+0.5*x^2}
\psplotImp(-11.0,-8.0)(14.0,7.0){-16.0+1.0*y^2+1.0*x^2}
\psplotImp(-11.0,-8.0)(14.0,7.0){-1.0+0.0625*y^2+0.08333333333333333*x^2}
\psplot[linecolor=lightgray]{-9.68}{13.66}{(--2.--0.4294487507354019*x)/1.024295039463181}
\psplot[linecolor=lightgray]{-9.68}{13.66}{(--3.930370166695838-0.42694233953616334*x)/-4.810276799130294}
\psplot[linecolor=lightgray]{-9.68}{13.66}{(-3.9303701666958424--1.4512373789993456*x)/-2.060903212636129}
\begin{scriptsize}
\psdots[dotsize=8pt 0,dotstyle=*,linecolor=darkgray](0.,-4.)
\rput[bl](-0.08,-4.68){\darkgray{$S$}}
\psdots[dotsize=8pt 0,dotstyle=*,linecolor=darkgray](0.,4.)
\rput[bl](-0.12,4.34){\darkgray{$N$}}
\psdots[dotsize=8pt 0,dotstyle=*,linecolor=darkgray](-1.374686793247079,-0.9390898592677543)
\rput[bl](-1.1,-0.52){\darkgray{$F$}}
\psdots[dotsize=8pt 0,dotstyle=*,linecolor=darkgray](1.3746867932470865,0.9390898592677545)
\rput[bl](1.46,1.26){\darkgray{$B$}}
\psdots[dotsize=8pt 0,dotstyle=*,linecolor=darkgray](-3.435590005883215,0.5121475197315905)
\rput[bl](-3.76,0.8){\darkgray{$O$}}
\psdots[dotsize=8pt 0,dotstyle=*,linecolor=darkgray](3.435590005883215,-0.512147519731591)
\rput[bl](3.76,-0.8){\darkgray{$\infty$}}
\psdots[dotsize=2pt 0,dotstyle=*,linecolor=darkgray](-0.040457868926124955,1.9355999417622334)
\rput[bl](0.04,2.06){\darkgray{$B'$}}
\psdots[dotsize=2pt 0,dotstyle=*,linecolor=darkgray](-8.379988211529422,-1.5608544482033058)
\rput[bl](-8.3,-1.44){\darkgray{$F'$}}
\end{scriptsize}
\end{pspicture*}
\end{figure}
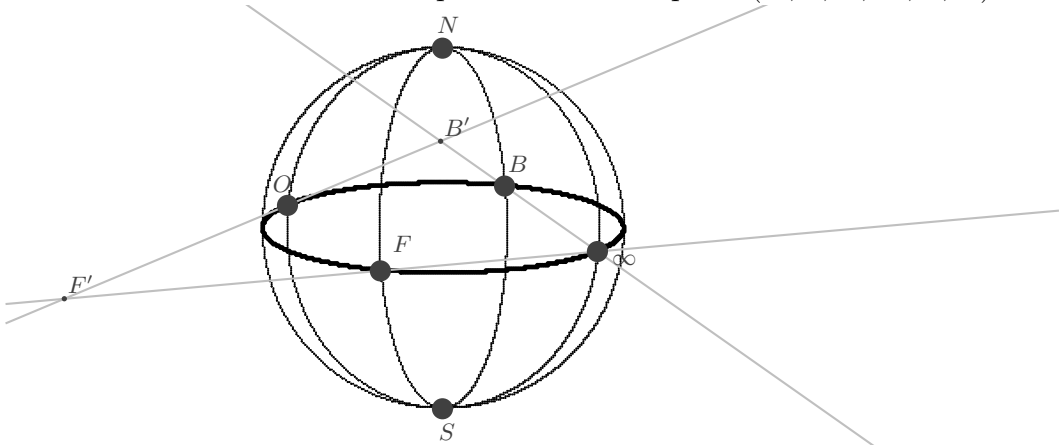
Its real form $\cR$ is the equator, a circle.
The Hermitian operators are just the reals $\R$ inside $\bC$,  the unitary group of $\bC$ is the circle, $\UU = S^1$,
and the classical complex Cayley transform takes the reals into $S^1$: there is just one point missing to complete the image. 
For sake of illustration, 
we'll do something foolish, namely represent the Riemann sphere  together with   4 poles called
{\em north, south, west, and east pole}, denoted by $N,S,W=\infty,O$.
Even more foolish, add another two poles, the {\em front pole} $F$ and the {\em back pole} $B$.
The poles $O,F,W,B$ lie on the {\em equator}, which is the horizontal
circle, model of the real projective line $\cR=\R\PP^1$, the completion of the
real line $\R = \Herm(\bC)$. By stereographic projection from  $\infty$,
the equator (taken out $\infty$)  is identified with the tangent
line at the equator at $O$ (the images of $F,B$ under this projection are denoted by $F',B'$; in fact, the whole 
sphere, taken out $\infty$, is identified with the tangent plane of $\cS$ at $O$; the points $N$ and $S$ then
correspond to $i F'$ and $-i B'$). 

\ssk
What are the relevant groups acting here, and by what are they determined?
The biggest relevant group is the {\em conformal group}  $\PP\SL(2,\bC)$, acting by conformal transformations of the sphere;
next comes its subgroup $\PP\SL(2,\R)$, fixing the real form given by the equator.
To single out a rotation group $\UU$, we may choose an arbitrary point $N$ of the sphere not belonging to the equator, call it ``north pole'', and call
its complex conjugate $S=N^*$  ``south pole''.
The data $(\cR; N,S)$ determine a unique rotation group $\UU= \PP\SO(2)\subset \PP \SL(2,\R)$ acting 
simply transitively on the equator $\cR$. 
The east and the west poles are of course not fixed under this action: they are used to define the linear picture
by stereographic projection, but they are not an invariant datum. 
All of this  generalizes to  ``completed quantum theory'' as defined in Part I: 
in order to describe the unitary group action, the pair $(N,S)$ seems to be a much more natural ``reference pair'' than
 $(0,\infty)$ (or any other pair of opposite points on the equator, like $(F,B)$). 
It is important that the invariant pair $(N,S)$ lies {\em outside} the real universe: in fact, 
in the chart given by stereographic projection, it is purely imaginary.  
Whenever we want to translate between ``usual'' and ``completed'' quantum theory, the following geometric picture will be useful: 
we consider the three pairs of vertices
$(N,S), (O,W), (F,B)$ as vertices of an octahedron. 
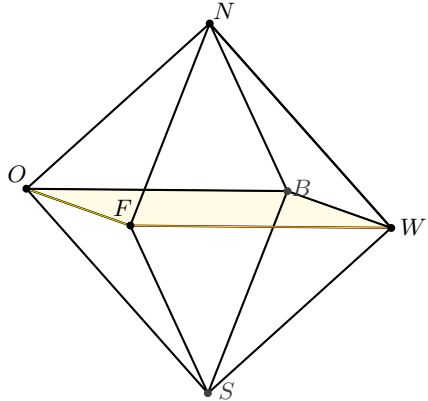
\begin{figure}[h]\caption{The octahedron $(N,S; O,W; F,B)$} \label{fig:octahedron}
\newrgbcolor{ttqqtt}{0.2 0. 0.2}
\newrgbcolor{ffdxqq}{1. 0.8431372549019608 0.}
\psset{xunit=0.6cm,yunit=0.5cm,algebraic=true,dimen=middle,dotstyle=o,dotsize=5pt 0,linewidth=0.8pt,arrowsize=3pt 2,arrowinset=0.25}
\begin{pspicture*}(-2.6636363636363605,-6.121818181818184)(13.096363636363648,4.998181818181821)
\pspolygon[linewidth=0.4pt,linecolor=ffdxqq,fillcolor=ffdxqq,fillstyle=solid,opacity=0.1](4.816363636363643,-1.3218181818181818)(2.516363636363642,-0.3418181818181814)(8.296363636363642,-0.40181818181818146)(10.596363636363643,-1.3818181818181818)
\psline(4.816363636363643,-1.3218181818181818)(2.516363636363642,-0.3418181818181814)
\psline[linecolor=ttqqtt](4.816363636363643,-1.3218181818181818)(10.596363636363643,-1.3818181818181818)
\psline(4.816363636363643,-1.3218181818181818)(6.576363636363643,4.05818181818182)
\psline(6.576363636363643,4.05818181818182)(2.516363636363642,-0.3418181818181814)
\psline(6.576363636363643,4.05818181818182)(10.596363636363643,-1.3818181818181818)
\psline(4.816363636363643,-1.3218181818181818)(6.536363636363643,-5.781818181818185)
\psline(8.296363636363642,-0.40181818181818146)(6.536363636363643,-5.781818181818185)
\psline(6.576363636363643,4.05818181818182)(8.296363636363642,-0.40181818181818146)
\psline[linewidth=0.4pt,linecolor=ffdxqq](4.816363636363643,-1.3218181818181818)(2.516363636363642,-0.3418181818181814)
\psline[linewidth=0.4pt,linecolor=ttqqtt](2.516363636363642,-0.3418181818181814)(8.296363636363642,-0.40181818181818146)
\psline[linewidth=0.4pt,linecolor=ttqqtt](8.296363636363642,-0.40181818181818146)(10.596363636363643,-1.3818181818181818)
\psline[linewidth=0.4pt,linecolor=ffdxqq](10.596363636363643,-1.3818181818181818)(4.816363636363643,-1.3218181818181818)
\psline(2.516363636363642,-0.3418181818181814)(6.536363636363643,-5.781818181818185)
\psline(10.596363636363643,-1.3818181818181818)(6.576363636363643,4.05818181818182)
\psline(10.596363636363643,-1.3818181818181818)(6.536363636363643,-5.781818181818185)
\psline(10.596363636363643,-1.3818181818181818)(8.296363636363642,-0.40181818181818146)
\psline(2.516363636363642,-0.3418181818181814)(8.296363636363642,-0.40181818181818146)
\begin{scriptsize}
\psdots[dotsize=3pt 0,dotstyle=*,linecolor=black](4.816363636363643,-1.3218181818181818)
\rput[bl](4.436363636363643,-1.0618181818181818){{$F$}}
\psdots[dotsize=3pt 0,dotstyle=*,linecolor=black](2.516363636363642,-0.3418181818181814)
\rput[bl](2.0963636363636417,-0.16181818181818136){{$O$}}
\psdots[dotsize=3pt 0,dotstyle=*,linecolor=black](10.596363636363643,-1.3818181818181818)
\rput[bl](10.796363636363644,-1.5618181818181818){{$W$}}
\psdots[dotsize=3pt 0,dotstyle=*,linecolor=black](6.576363636363643,4.05818181818182)
\rput[bl](6.656363636363643,4.17818181818182){{$N$}}
\psdots[dotsize=3pt 0,dotstyle=*,linecolor=darkgray](8.296363636363642,-0.40181818181818146)
\rput[bl](8.396363636363644,-0.5218181818181816){\darkgray{$B$}}
\psdots[dotsize=3pt 0,dotstyle=*,linecolor=darkgray](6.536363636363643,-5.781818181818185)
\rput[bl](6.77,-5.941818181818183){\darkgray{$S$}}
\end{scriptsize}
\end{pspicture*}
\end{figure}
To this octahedron one associates a natural symmetry group, the {\em octahedron group}, having 48 elements. Half of them are realized by 
{\em holomorphic} transformations of $\cS$, the other half by {\em antiholomorphic} transformations (Theorem \ref{th:octahedral}).
Among the holomorphic transformations are the famous {\em Cayley transforms}, which are holomorphic maps of order $3$, corresponding to
rotating, eg., the triangles $FNW$ and $OBS$ 
by $2\pi / 3$. They are the key ingredient for a more careful analysis of the geometric situation: for completed quantum mechanics, the most important
feature of the unitary group is that it is ``a union of affine spaces''
(Theorem \ref{th:U-complete}).  This result relies on the ``positivity'' assumption on the Jordan algebra. 
-- Summing up, we have the choice
\begin{itemize}
\item
(linear quantum theory) 
to consider the east, west, front and back poles as intrinsic data, fixed under $\UU$; then all six poles are fixed, and we get
the (co) adjoint action of $\UU$ on its Lie algebra, 
\item
or (completed quantum theory) 
 to consider only the north and south poles as intrinsic data; then we get the simply transitive action of 
$\UU$ by left or right translations on itself.
\end{itemize}
Let's stress again that these pictures are mathematically equivalent.
In our toy example, the linear  Schr\"odinger action is not visible: since $\UU = \UU(1)$ is abelian, the adjoint and coadjoint actions are trivial, translating the fact that
global right and left action coincide. 
But as soon as $\UU$ becomes non-abelian, the ``quantum features'' become visible. 
This does of course not mean that the new picture must be the final word: it seems very well possible that in order to describe 
``wave function collapse'' we still need another picture (Part III ?). 

\subsection{Planck's constant,  differential calculus, and quantum theory}
Although Lie theory, even infinite dimensional, may be considered as ``standard'', 
the presence of Planck's constant in the time evolution formula (\ref{eqn:Sch})
  indicates that not everything is ``business as usual''.
This constant also shows up as {\em Jordan-Lie constant} in the crucial property (JL4) of a Jordan-Lie algebra
(section \ref{sec:JL}).   
 As far as I understand the situation, from the point of view of ``conceptual calculus'' (\cite{Be18, Bexy}),
even in classical calculus, the identification of a vector space $V$ with its tangent space $T_0 V$ is not as
``canonical'' as one usually tends to think --
see Subsection \ref{sec:q-calculus} for this issue. 
The non-linear approach to quantum mechanics offers the possibility to distinguish the levels of ``space'' and ``tangent
space'', and to interprete Planck's constant as a factor showing up each time we identify ``space'' with ``tangent space''.
Put differently, the ``locally linear'' nature of quantum geometry forces us to identify sometimes a ``global'' geometric
object with an ``infinitesimal'' one, and if we do so, a constant $\hbar$ shows up.
This touches foundational questions of differential calculus -- 
we hope to be able to say  more on this at another occasion (\cite{Bexy}).

\section{Jordan-Lie algebras}\label{sec:JL}

This first section is purely algebraic, dealing with algebras.  By {\em algebra (over a commutative ring $\K$)} we mean a $\K$-module together 
 with a $\K$-bilinear product $\beta:\bA \times \bA\to \bA$. We shall often write also $\beta(x,y)=xy$.

\subsection{The associator}
 The  {\em associator} of an algebra $(\bA,\beta)$ is defined by
\begin{equation}
A_\beta(x,y,z) := (xy)z - x(yz) = \beta(\beta(x,y),z) - \beta(x,\beta(y,z)) .
\end{equation}
By definition, an algebra $(\bA,\beta)$ is {\em associative} iff $A_\beta = 0$. 
The associator depends quadratically on the product $\beta$, that is, 
$A_{k\beta} = k^2 A_\beta$, for any $k \in \K$.
For every algebra $(\bA,\beta)$, we  define the {\em symmetric} and {\em skew-symmetric part}  by
\begin{align*}
J(x,y) & = xy+yx = \beta(x,y) + \beta(y,x), \\
L(x,y) & =xy-yx=\beta(x,y) - \beta(y,x) .
\end{align*}
These are algebra structures on $\bA$, such
that $\beta(x,y) = \frac{1}{2}(J(x,y)+L(x,y))$.

\begin{lemma}\label{la:JL!!} 
Let $\bA$ be an {\em associative} algebra, i.e.,
$A_\beta = 0$. Then the associatiors of the symmetric and of the skew-symmetric part agree, up to a sign:
 $$
 A_J = - A_L . 
 $$ 
\end{lemma}

\begin{proof}
By assumption, the product $\alpha(x,y)=xy$ is associative. Using this,
\begin{align*}
L(L(x,y),z) - L(x,L(y,z)) 
 & =  (xy-yx)z - z(xy-yx) - x(yz-zy) + (yz - zy)x \\
& = xzy + yzx - yxz - zxy ,
\\
J(J(x,y),z) - J(x,J(y,z))  & =
  (xy + yx)z + z(xy+yz) - x(yz+zy)- (yz+zy)x
\\
&=  yxz + zxy - xzy - yzx,
\end{align*}
which is the {\em negative} of the preceding expression, whence $A_J=-A_L$. 
\end{proof}

\begin{remark}
See Appendix \ref{app:LJ} for an interpretation, in terms of ternary products,  of the quantity
$A_L = - A_J = R_T$.
\end{remark}

\nin
Recall  that, if $\beta$ is associative,  $L$ is a {\em Lie bracket} (it is skew and satisfies the Jacobi identity)
and $J$ a {\em Jordan algebra product} (it is commutative and satisfies the Jordan identity). 
Then we often write $[x,y]=L(x,y)$ and $x \bullet y = J(x,y)$.

\subsection{Jordan-Lie algebras}
The following definition and results generalize those 
given in the litterature, which  concern the special case corresponding to a
$C^*$-algebra, see eg., \cite{E, L98},  and the
\href{https://ncatlab.org/nlab/show/Jordan-Lie-Banach+algebra}{one on the n-lab}, 
where the Jordan-Lie constant is implicitly supposed to be $\kk=1$.

\begin{definition}
Let $\kk \in \K$ be a constant.
A $\K$-module $V$, equipped with two bilinear products denoted by $[x,y]$ and
 $x \bullet y$ is called a {\em Jordan-Lie algebra with Jordan-Lie constant $\kk$}
if the following holds:
\begin{enumerate}
\item[(JL1)] $(V,[ \cdot,\cdot ])$ is a Lie algebra, i.e., it is skew and satisfies the Jacobi-identity,
\item[(JL2)] $(V,\bullet)$  is commutative,
\item[(JL3)] the Lie algebra acts by derivations of $\bullet$,  that is,
$$
[x,u \bullet v] = [x,u] \bullet v + u \bullet [x,v] ,
$$
\item[(JL4)] the {\em associator identity}: 
associators of both products are proportional, by a factor $\kk$,
that is, $A_\bullet = k A_{[-,-]}$. Written out, this reads
$$
(x \bullet y) \bullet z - x \bullet (y \bullet z) = \kk \bigl(
[[x,y],z]-[x,[y,z]] \bigr),
$$
or, by using the Jacobi-identity (JL1), this can also be written
$$
(x \bullet y) \bullet z - x \bullet (y \bullet z) = - \kk  [[z,x],y] = \kk  [[x,z],y].
$$
\end{enumerate}
A {\em morphism} of Jordan-Lie algebras is a linear map that is both a morphism of
$\bullet$ and of $[-,-]$.
\end{definition}

\begin{lemma}
Under the preceding conditions, the algebra $(V,\bullet)$ is a Jordan algebra.
\end{lemma}

\begin{proof}
We have to prove that the Jordan identity
$(x\bullet y)\bullet x^2 - x\bullet (y \bullet x^2)=0$ holds, where $x^2=x\bullet x$.
From (JL3) with $u=v=x$ we infer $[x,x^2]=2 [x,x] \bullet x = 0$. 
Letting  $z=x^2$ in the second display of (JL4),
the Jordan identity now follows.
\end{proof}

\begin{remark}
When $\kk=0$, the product $\bullet$ is associative, by (JL4),  and commutative, by (JL2), hence in this case
we get the definition of a
\href{https://en.wikipedia.org/wiki/Poisson_algebra}{\em commutative Poisson algebra}.
\end{remark}

\subsection{The case of negative Jordan-Lie constant}
According to Lemma \ref{la:JL!!}, every associative algebra $(\bA,\beta)$ gives rise to a Jordan-Lie algebra
$(\bA,J,L)$ with constant $\kk = -1$. 
More generally, this holds whenever $-\kk$ is a square in $\K$ (so, if $\K=\R$, when $\kk$ is negative):

\begin{theorem}\label{th:JL1}
Let $\bA$ be an associative algebra over $\K$, and let $u,w \in \K^\times$. Then $\bA$ with products
$$
a \bullet b =w(ab+ba), \qquad
[a,b]_u = u(ab-ba),
$$
becomes a Jordan-Lie algebra over $\K$ with Jordan-Lie constant $\kk = - \frac{w^2}{ u^2} $. 
Conversely, assume $(V,\bullet,[--])$ is a Jordan-Lie algebra with Jordan-Lie constant $\kk$ such that
$- \kk$ is a square in the base field $\K$: $\exists c \in \K^\times$,
$-\kk = c^2$. Choose $u,w \in \K^\times$ such that $c = \frac{w}{u}$ and let
$$
ab := \frac{1}{2w} a \bullet b + \frac{1}{2u} [a,b].
$$
Then this defines an associative product on $V$, and both constructions are inverse to each other.
For fixed values of $\kk,u,w$, this defines an equivalence of categories between  associative algebras and Jordan-Lie
algebras with Jordan-Lie  constant $\kk$.
\end{theorem}

\begin{proof}
If $\bA$ is associative, then the symmetric part is a Jordan, and the skew-symmetric part a Lie product, whence
(JL1) and (JL2).  To prove (JL3), we compute 
$$
[x,yz] = u(xyz - yzx) = u(xyz - yxz + yxz - yzx )
= [x,y]z + y[x,z] ,
$$
which means that the Lie algebra acts by derivation of the associative product, and hence also by derivations
of its symmetric and skew-symmetric parts.
Property (JL4) follows 
with $J=\bullet$ and $L=[-,-]$: since $A_j = - A_L$ by Lemma \ref{la:JL!!},
$$
A_{wJ}=w^2 A_J = - w^2 A_L = - \frac{w^2}{u^2} A_{uL} = \kk  \, A_{uL}.
$$
Note also that the associative product is recovered via
$$
xy = \frac{1}{2w} w(xy+yx) + \frac{1}{2u} u (xy-yx) = 
\frac{1}{2w} x\bullet y  + \frac{1}{2u} [x,y].
$$
To prove the converse, define $ab$ as in the claim, and compute the associator of this product.
There are 8 terms: 
\begin{align*}
(xy)z-x(yz) & = \frac{1}{4w^2} (x\bullet y)\bullet z +
 \frac{1}{4uw}( [x,y]\bullet z + [x\bullet y,z]) +  \frac{1}{4u^2} [[x,y],z] - 
 \\
& \quad \qquad \bigl(
 \frac{1}{4w^2} x\bullet (y \bullet z) + 
  \frac{1}{4uw} ([x,y\bullet z]+x\bullet [y,z]) + 
   \frac{1}{4u^2} [x,[y,z]] \bigr)
   \\
 &=  \frac{1}{4w^2} \bigl( x\bullet y)\bullet z - x\bullet (y \bullet z) \bigr) +
   \frac{1}{4u^2} \bigl( [[x,y],z] - [x,[y,z]] \bigr)
 \\
 & = (  \frac{1}{4w^2} \kk +  \frac{1}{4u^2} ) ([x,y]z - [x,[y,z]]) = 0 ,
\end{align*}
where  the 4 ``mixed terms'' cancel out because of (JL1,2,3) (second equality), and the remaining 4 terms give zero because of
(JL4) (third and fourth equality).
Thus the product $ab$ is associative, and as noticed above, both constructions are inverse to each other.
It is straightforward that morphisms of associative algebras are morphisms of the Lie and Jordan products, and
conversely, if a linear map is both a Lie and Jordan algebra morphism, it will also be a morphism of the product $ab$,
hence we get an equivalence of categories.
\end{proof}

\begin{proposition}
With notation as in the theorem, the following are equivalent:
\begin{enumerate}
\item
 $e$ is a unit for the associative product,
 \item
  $\frac{1}{2w} e$ is a unit for the Jordan product,
 and $[e,a]=0$ for all $a\in V$.
\end{enumerate} \end{proposition}

\begin{proof}
Direct from the formulae given in the theorem.
\end{proof}

\subsection{The case of positive Jordan-Lie constant}
Now, what about Jordan-Lie algebras with Jordan-Lie constant $\kk = + 1$ (in the real case, positive $\kk$) ?
According to Theorem \ref{th:JL1}, they can be realized whenever there is an element $i \in \K$ such that $i^2 = -1$, by
choosing [$w=1$ and $u=i$], or [$w=i$ and $u=1$].
For convenience, the following  is stated for $\K=\R$, but it extends to any  base ring, cf.\  Remark \ref{rk:anyK}. 

\begin{theorem}\label{th:JL2}
Assume $(\bA,*)$ is a complex $*$-algebra, and let $v,w \in  \R$ be non-zero.
Then $V=\Herm(\bA) = \{ x \in \bA \mid a^* = a \}$ with products
$$
a \bullet b =w(ab+ba), \qquad
[a,b]_{iv} = i v(ab-ba),
$$
becomes a real Jordan-Lie algebra with Jordan-Lie constant $\kk =  \frac{w^2}{ v^2}$. 
Conversely, assume $(V,\bullet,[--])$ is a real Jordan-Lie algebra with Jordan-Lie constant $\kk >0$, and
choose $u,w \in \R^\times$ such that $\kk = \frac{w^2}{u^2}$.
We extend $\bullet$ and $[-,-]$ by $\bC$-bilinearity to complex products on the complexified vector space
$\bA := V_\bC = V \oplus i V$, 
 and let
$$
ab := \frac{1}{2w} a \bullet b + \frac{1}{2iv} [a,b].
$$
Then this defines a complex associative product on $\bA$, and complex conjugation defines an involution on $\bA$
turning it into a $*$-algebra. 
Again,
both constructions are inverse to each other.
In particular, fixing the choice $v=w=1$, we get an equivalence of categories between 
complex  $*$-algebras,  and real Jordan-Lie
algebras with Jordan-Lie constant $\kk = 1$.
\end{theorem}

\begin{proof}
If $\bA$ is a complex $*$-algebra, apply the preceding theorem for real $w$ and imaginary $u = iv$. 
This gives a complex Jordan-Lie algebra with $\kk =-  \frac{w^2}{(iv)^2} = \frac{w^2}{v^2} >0$.
Moreover, $\Herm(\bA)$ is stable both under the Jordan and under the Lie product, hence it is a
Jordan-Lie subalgebra of $\bA$, with the same $\kk$.
For the proof of the converse, by general algebra, the complexification of a Jordan-Lie algebra is
again a Jordan-Lie algebra, with complex conjugation being an automorphism of both products.
Applying the preceding theorem, we recover the complex associative product, with constant
$\kk = - \frac{w^2}{(iv)^2} = \frac{w^2}{v^2} >0$.
It remains to prove that complex conjugation is an anti-automorphism of this associative product:
\begin{align*}
(ab)^*  & =  \frac{1}{2w} (a \bullet b)^* - \frac{1}{2iv} [a,b]^* = 
 \frac{1}{2w} a^* \bullet b^* - \frac{1}{2iv} [a^*,b^*]
 \\
 &=  \frac{1}{2w} b^* \bullet a^* + \frac{1}{2iv} [b^*,a^*] = b^* a^* .
\end{align*}
Finally, the arguments establishing the equivalence of categories follow from general algebra, as in
the preceding theorem.
\end{proof}

\begin{remark}\label{rk:anyK}
Working over general base fields or rings $\K$ instead of $\R$, there is an analog of the theorem with $\bC$ replaced by
the ring $R:=\K[X]/(X^2 + \kk) = \K \oplus   {\it j}   \,   \K$ with  $j^2 = \kk$. 
Namely (fixing the choice $w=\frac{1}{2}$), the Jordan-Lie algebra gives rise to an associative product on $V_R = V \oplus j V$, given by
$$
ab = a \bullet b + \frac{j}{2} [a,b] 
$$
where $\bullet$ and $[-,-]$ are the $R$-bilinear extensions of the original products onto $V_R$.  This
product is associative by Theorem \ref{th:JL1}.
As in the proof of Th.\ \ref{th:JL2}, it follows that {\em every} Jordan-Lie algebra with given $\kk$ is obtained as the $1$-eigenspace
of an involution in an involutive associative algebra.
In particular, these arguments work when $\kk=0$: every commutative Poisson algebra is an eigenspace of an 
associative, not necessarily commutative algebra, which in this case is constructed by using the algebra of 
{\em dual numbers} ($j^2 = 0$).
\end{remark}

\subsection{Positive Jordan-Lie algebras}
One should not mix up ``positivity'' of the Jordan-Lie constant $\kk$ 
 with ``positivity conditions'' on the algebra $\Herm(\bA)$: 
these two things are independent of each other. More precisely, positivity of $\kk$ is a necessary, but by no means sufficient condition
for    $\Herm(\bA)$ to be
``positive''  in the sense of ordered Jordan algebras (see Example \ref{exa:ordered}). 
However, most authors implicitly add a ``positivity'' condition in their definitions, since they aim at $C^*$-algebras.
Here is my version of such a positivity condition: 

\begin{definition}
A {\em $P^*$-Jordan-Lie algebra} is a  Jordan-Lie algebra such that:
\begin{enumerate}
\item[(1)]
$\K$ is an ordered field, and the Jordan-Lie constant is positive:
$k>0$,
\item[(2)]
the Jordan algebra $(V,\bullet)$ is an {\em ordered Jordan algebra} (cf.\  Appendix A of Part I), that is, 
$V$ is an ordered $\K$-module, and its positive cone $\Omega = \{ x \in V\mid x>0 \}$ satisfies
$\Omega \subset V^\times$, and: 
$\forall a \in \Omega, \forall b \in V^\times$,
$aba \in \Omega$,
\item[(3)] 
$1 + aba \in \Omega$ for all $a \in V$ and $b\in \Omega$.
\end{enumerate}
\end{definition}

\nin
For instance, if the involutive associative algebra corresponding to a Jordan-Lie algebra with $\kk >0$ is a
$C^*$-algebra, then the above conditions are fulfilled (but the converse is not true).
In finite dimension over $\K=\R$, the $P^*$-condition implies that $\Omega$ is {\em of non-compact type}, and
hence the Lie algebra {\em is of compact type}, thus we end up with the case of the compact group $\UU(n)$ from the following

\begin{example}\label{exa:ordered}
 The $*$-algebra $\bA = M(n,n;\bC)$ with involution
$a^* = I_{p,q} \overline a^t I_{p,q}$, where $I_{p,q}$ is the diagonal matrix having $p$ coefficients $1$ and
$q$ coefficients $-1$, is ordered (i.e., a $P^*$-algebra) if, and only if, $p=0$ or $q=0$ (iff it is a $C^*$-algebra).
These cases are distinguished: the Jordan algebra $\Herm(p,q;\bC)$ of Hermitian matrices of signature
$(p,q)$ is Euclidean iff $p=0$ or $q=0$, iff the pseudo-unitary group $\UU(p,q)$ is compact. 
\end{example}

\begin{remark}[Physics constants: sign of $\hbar$]\label{rk:hbar}
In physics contexts, $\kk$ is positive, and we let $\hbar = 2 \sqrt{\kk}$ (positive square root). 
We fix $w=\frac{1}{2}$, so
$u=\frac{i}{\hbar}$, and
$\kk = \frac{\hbar^2}{4}$, and
\begin{equation}
ab = a\bullet b + \frac{\hbar}{2i} [a,b].
\end{equation}
We assume that $\hbar >0$, but  note that the opposite choice $-\hbar$ is related with the opposite product $ba$, leading to
the {\em same}  constant $\kk = \frac{1}{4 \hbar^2}$.  Thus the sign of $\hbar$ seems to be some kind of convention, corresponding to (implicit) conventions
of preferring left to right actions, or to write function symbols at the left of their arguments.  
\end{remark}

\subsection{Summary}
The setting of positive ($P^*)$-Jordan-Lie algebras is mathematically equivalent to the setting of positive $*$-algebras, that is, to the setting
of quantum mechanics (Part I). It is  a ``purely real'' setting, but
complex numbers come out of the assumption that the Jordan-Lie constant be {\em positive}.

\section{Geometry of $\bA$-unitary groups}\label{sec:algeometry}

\subsection{The general algeometry problem}
Let us call ``general algeometry problem'' the following: 
by the theory of Sophus Lie, we know that {\em Lie groups} correspond to (finite dimensional, real) {\em Lie algebras}. We ask: 
{\em
What can one say for other classes of algebras: given a class of algebras, defined by certain algebraic identities, is there a class of 
 global, geometric objects ``integrating''
such algebras?} \footnote{In various papers, and on my homepage, I've called this the ``general coquecigrue problem''. But possibly, by now time has come to
give it a more serious name.} 

\ssk
I've been working for quite a long time
on this kind of questions, mostly for certain classes of algebras, namely for {\em Jordan algebras} and {\em associative algebras},
as well as for their ternary algebraic analogs. As mentioned in Part I, Jordan algebras correspond to so-called {\em Jordan geometries}, or {\em 
generalized projective geometries}. Since Lie algebras correspond to Lie groups, this suggests that
for Jordan-Lie algebras,  the geometric object ought to be a space carrying two kinds of structure:
\begin{enumerate}
\item
a {\em Lie group structure} (or {\em Lie torsor structure}, if we don't want to fix the unit element),
\item
some kind of projective structure: a {\em generalized projective line} (see Part I).
\end{enumerate}
\nin
It is quite clear that (JL3) translates by saying that the Lie group (1) shall act by automorphisms of the projective structure (2). 
The difficult part is:  {\em  how to translate, geometrically, the compatibility condition (JL4)?}
Anywhow, this geometric object shall correspond to what we will call below (Definition \ref{def:U-setting}) the {\em unitary setting of
completed quantum mechanics}. 
Let's say already here that our answer  is still far from definitive (cf.\ Appendix \ref{sec:JL-geometry}). In particular,   it does not (yet) give a direct clue how to interprete the
``measurement problem'' in the geometric setting.

\subsection{Imbedding of unitary groups into Lagrangians} 
General unitary, orthogonal and symplectic groups can be imbedded into varieties of Lagrangian subspaces. 
The basic idea is simple and fairly well-known: 
 a linear map
$g:V \to W$ preserves a bi- or sequilinear form $\beta$, that is,
$\beta (gv, gw) = \beta(v,w)$, iff  its {\em graph} is a Lagrangian subspace for the form $B$ on $W \oplus V$, 
sometimes denoted by $\beta \ominus \beta$,  given by
\begin{equation}
B(( v,v'),(w,w')) := \beta(v,w) - \beta(v',w') .
\end{equation}
Indeed, the graph of $g:V \to W$ is the set $\Graph_g =\{ (gx,x) \mid x \in V \} \subset W \oplus V$ (which is a linear subspace if $g$ is linear), and
$\Graph_g$ is Lagrangian for $\beta \ominus \beta$: 
\begin{equation}
B((gx,x),(gy,y)) =
\beta (gx, gy) - \beta(x,y) = 0 .
\end{equation}
 Thus via $g \mapsto \Graph_g$,
 unitary groups can be imbedded as subsets into Grassmannian or Lagrangian varieties.  By a compacity argument, if the group is compact, the imbedding
 is onto. 
 We shall specialize this general construction to the case we are interested in. First, we define the relevant unitary groups.

\subsection{The $\bA$-unitary groups}
Recall from Part I the setting of completed quantum theory, given by a
 $*$-algebra $\bA$ (just an associative $\K$-algebra with involution $*$), and the associated Hermitian projective
line $\cR$, a real form of the projective line $\cS = \bA \PP^1$.
The automorphism group of $\cS$ corresponds to the projective group $\PP \Gl(2,\bA)$.
In the language of completed quantum theory, the algeometry problem for Jordan-Lie algebras raises the following question:
what intrinsic geometric datum allows to identify $\cR$ with a unitary group (or, more correctly, unitary torsor)?  The group structure is an additional
structure on $\cR$, which was not present in the setting of Part I.
Let's start by defining the unitary groups:

\begin{definition} \label{def:U(A)}
Let $\bA$ be an associative $*$-algebra $\bA$, and $M$ a real invertible $n\times n$-matrix.
The {\em group of $M$-unitary $(n\times n)$-matrices with coefficients in
$\bA$} is 
\begin{equation}\label{eqn:UnM}
 \UU(M,\bA,*):=\UU(M, \bA) := \{ A \in M(n,n;\bA)  \mid A^* M A = M  = A  M  A^*\} 
\end{equation}
where $(A^*)_{ij} = (a_{ji})^*$ is the conjugate-transposed matrix (with respect to the involution $*$ of $\bA$). 
Since $(AB)^* = B^* A^*$ and $1^* = 1$, this is indeed a group. When $M = 1_n$ is the unit matrix, we also write
\begin{equation}\label{eqn:Un}
 \UU(n,\bA,*):=\UU(n, \bA) := \{ A \in M(n,n;\bA)  \mid A^* A = 1_n = A  A^*\} 
\end{equation}
We are most notably interested in the cases $n=2$ and $n=1$: in the latter, we get the {\em unitary group of $\bA$} 
\begin{equation}\label{eqn:U}
\UU := \UU(\bA):=\UU(1, \bA) := \{ x \in \bA \mid x^* x = 1 = x x^*  \} .
\end{equation}
\end{definition}

Note that $M(n,n;\bA)$ is again a $*$-algebra, and we have $\UU(n,\bA) = \UU(1,M(n,n;\bA))$.
When $\bA = \bC$, then the groups $\UU(n)$ are {\em compact}. This, of course, does not carry over to the general case,
 not even to the case of $C^*$-algebras.
The main property of the ``compact-like'' groups, replacing compactness in arbitrary dimension, will be stated in Theorem \ref{th:U-complete}.

 \subsection{The unitary setting of completed quantum mechanics}
In the following, we will introduce a slight, but important shift in the setting: 
rather than by $(\cS,\tau)$, the setting should be determined by $(\cS; \tau;(N,S))$, where $(N,S)$ is a pair of points called
{\em north and south pole}.  
Before fixing a particular pair, let's explain that fixing  any pair or triple of 
 points on $\cS = \bA\PP^1$  defines certain geometric structures:
 \begin{enumerate}
 \item
 a {\em transversal pair} $(a,b)$ (recall: $a,b$ are {\em transversal} if their sum is direct: 
  $\bA^2 = a \oplus b$) defines certain
  holomorphic maps $\cS \to \cS$,
 \item
 a  {\em transversal triple} $(a,b,c)$ (meaning $a,b,c$ are pairwise transversal) defines certain other
 holomorphic and antiholomorphic maps, such that
 \item
 when the triple $(a,b,c)$ is
 completed to a  $6$-tuple $(a,b;c,d;n,s)$, there is a  group of holomorphic and
 antiholomorphic transformations generated by these maps; the group  turns out to be an {\em octahedron group}.
 \end{enumerate}
 \nin
 We explain items (1), (2), (3) in this order:
 
 \subsubsection{Holomorphic automorphisms defined by a pair of points}
 Given a  {\em transversal pair} $(a,b)$, and $\lambda \in \bC^\times$, 
we consider the linear map that is given by the ``matrix'' 
$\bigl( \begin{smallmatrix} \lambda & 0 \\0& 1 \end{smallmatrix} \bigr)$ with respect to the decomposition
$\bA^2 = a \oplus b$. On $\cS = ±bA \PP^1$ this induces a holomorphic diffeomorphism 
\begin{equation}
\lambda_{a,b}:\cS\to \cS,  \quad [a r + bs] \mapsto [ \lambda ar +  bs] 
\end{equation}

\subsubsection{Holomorphic automorphisms and antiholomorphic antiautomorphisms defined by a triple of points}
 Fixing  
  a {\em transversal triple} $(a,b,c)$ means to fix a common complement $c$ of $a$ and of $b$, for $(a,b)$ a transversal pair.
 Then $c$ can be considered as
 {\em diagonal} in the decomposition $\bA^2 = a \oplus b$, and hence serves to identify $a$ with $b$. Lets denote this situation by the
 notation
 $
 \bA^2 = a \oplus_c b \cong a \oplus a $.
Then  any other common complement of $a$ and $b$ can be identified with the graph of an $\bA$-linear isomorphism from $a$ to $a$;
that is, 
the set $U_{ab} = U_a \cap U_b$ of common complements
carries a group structure with neutral element $c$ and isomorphic to the group $\bA^\times$ (see \cite{BeKi} for more on this).

\begin{definition}
Given a transversal triple $(a,b,c)$, we define a {\em holomorphic symmetry} $J^{ab}_c:\cS \to \cS$, and an
{\em antiholomorphic reflection} $\tau^{ab}_c:\cS \to \cS$, as follows. 
\begin{enumerate}
\item
The map $J^{ab}_c$ is induced by the linear map
${\rm J}:a \oplus a \to a \oplus a$, $(u,v) \mapsto (v,u)$ (reflection at the diagonal $c$). 
Put differently, the map
$J^{ab}_c : \cS \to \cS$
is the  {\em central symmetry at the midpoint of $(a,b)$ in the affine space $U_c$}, or, yet differently, it is the inversion map
in the group $(U_{ab},c)$ (see \cite{Be14} for more on these ``inversions'').
\item
For $z \in \cS$, we define $\tau^{ab}_c(z) := z^{\perp,\beta}$ to be the orthocomplement of $z$ with respect to
the ``hyperbolic'' Hermitian form given on $\bA^2 = a \oplus b = a \oplus a$ by
\begin{equation}
\beta  ((u,v),(u',v')) = u^* u' - v^* v' .
\end{equation}
\end{enumerate}
\end{definition}

\nin
The form $\beta$ depends on the bases of $a$ and $b$, but the orthocomplement $z^{\perp,\beta}$ does not.
Both maps are bijections of order 2, fixing $c$ and exchanging $a,b$: 
\begin{align}
J^{ab}_c(a)=b, \qquad J^{ab}_c(b)=a, \qquad J^{ab}_c(c)=c ,\qquad J^{ab}_c \circ J^{ab}_c = \id_\cS ,
\\
\tau^{ab}_c(a)=b, \qquad \tau^{ab}_c(b)=a, \qquad \tau^{ab}_c(c)=c ,\qquad \tau^{ab}_c \circ \tau^{ab}_c = \id_\cS .
\end{align} 
The maps $J^{ab}_c$ are holomorphic, and the fixed point $c$ is an isolated fixed point -- 
in \cite{Be14}, we have characterized {\em Jordan geometries} by geometric properties of such sets of symmetries. 
The maps $\tau^{ab}_c$ are antiholomorphic, and
 the fixed point set of $\tau^{ab}_c$  is a {\em real form} of the complex space $\cS$.

\begin{theorem}\label{th:UU1}
Given a transversal triple $(a,b,c)$ in $\cS$,  let $\tau := \tau^{ab}_c$. Then
 the fixed point set 
$\cS^\tau$ is the set of $\beta$-Lagrangian subspaces of $\cS$, and its  subset
$$
U_{ab}^c := U_{ab} \cap \cS^\tau = \{ x \in \cS \mid \tau(x)=x, \, x \top a, x \top b \}
$$
is a subgroup of $U_{ab}$ with unit element $c$.  This group
can be identified with the $\bA$-unitary group $\UU(\bA,*)$ given by (\ref{eqn:U}), 
via the imbedding (defined with respect to the decomposition $\bA^2 = a \oplus b$)
$$
\UU(\bA,*)  \to \cS^\tau, \quad x  \mapsto [(1,x)] .
$$
\end{theorem}

\begin{proof}
In \cite{BeKi2} this result is proved by showing that $\tau$ is an {\em antiautomorphism} of the structure map $\Gamma$. 
Without using that general theory, 
let us check the statements here by direct computation: first of all, $z$ is a fixed point of $\tau^{ab}_c$ iff
$z^\perp = z$, iff $z$ is Lagrangian. Next, assume $z = z^\perp$ and $z \top a$, $z \top b$, so
$z=[(1,x)]$ with $x \in \bA$ (for the fixed decomposition $\bA^2 = a \oplus b$).
Then $z$ is Lagrangian for $\beta$ iff
$$
0=\beta ((1,x) ,(1,x) ) = 1 - x^* x ,
$$
i.e., iff $x^* x =1$, 
 iff $x \in \UU(\bA,*)$.
Hence  $U_{ab} \cap \cS^\tau$ is the imbedded group $\UU(\bA,*)$.
\end{proof}

Under the action of $\GL(2,\bA)$,
 every transversal triple $(a,b,c)$ is conjugate to the standard transversal triple
$(O,W,F)$,
\begin{equation}\label{eqn:OWF}
\begin{matrix}
O & = & [(1,0)] & = & \bA \oplus 0 , \\  
W  & = & [(0,1)]&= & 0\oplus \bA,  \\
F  & = & [(1,1)] &= & \{ (a,a)\mid a \in \bA \}
\end{matrix}
\end{equation}
(first and second factor of $\bA^2$, and 
the diagonal). 
We then define three other points by 
$N = i F$, $B = i^2 F$, $S = i^3 F$, where $i = i_{O,W}$ is the dilation operator by $i$ with respect to the decomposition
$\bA^2 = O + W$. 
Thus our transversal triple gives rise to
six ``poles'' of $\cS$, called {\em east, west, north, south, front, and back},  given by 
 \begin{align}
O:= [(1,0)], & \qquad W := [(0,1)], \label{H}
\\
N:= [(1,i)], & \qquad S := [(1,-i)], \label{V}
\\
F:=[(1,1)],& \qquad B :=[(1,-1)]. \label{D}
\end{align}
In the usual chart $\bA$, this corresponds to
$(0,\infty;i,-i;1,-1)$. 
The $6$-tuple of poles $(O,W;N,S;F,B)$ comes with a partition into three parts:
we call
(\ref{H}) {\em horizontal pair},
(\ref{V}) {\em vertical pair},
(\ref{D}) {\em depth pair}, and we say that $N$ is the {\em opposite of $S$}, and so on.
We represent the six poles by the six vertices of a regular octohedron, such that poles from the same part are represented by 
opposite vertices (Figure \ref{fig:octahedron}). 
Compare this with Figure \ref{fig:RS}, where the $8$ vertices of the octahedron are placed on the Riemann sphere.

\begin{definition}\label{def:U-setting}
We consider henceforth $(\cS,\tau;N,S)$ as fixed data describing the setting of completed quantum mechanics, and
call this the {\em unitary setting of completed quantum mechanics}. As in Part I,
the {\em principal real form of $\cS$} is
$\tau := \tau^{NS}_0$, respectively its fixed point set
$\cR := \cS^{\tau}$.
The group $\UU := \UU^{SN}_O$ defined in Theorem \ref{th:UU1} with unit $O$
 is called
the {\em unitary group of completed quantum mechanics} (however, the point $O$ is arbitrary and not part of the
setting). As a set,
$$
\UU = \cR \cap U_{NS} = \cS^\tau \cap U_{NS} = 
 \{ z \in \cS\mid \tau(z) = z,  z \top N, z\top S\} 
$$
(since $\tau(S)=N$ and $\tau$ preserves transversality, the last condition is redundant).
Forgetting the unit $O$ of this group, we get a torsor, called {\em the unitary torsor of completed quantum mechanics}
(simply transitive action of $\UU(\bA,*)$ on $\UU$). 
\end{definition}

\begin{theorem}
When $\bA = M(n,n;\bC)$, with its usual involution $x^* = \overline x^t$,
 then $\UU = \cR$, that is,
 $U^O_{SN} = \cS^\tau \cong \UU(n)$; in other words,  the imbedding $\UU\to \cS^\tau$  from Theorem \ref{th:UU1} is a bijection.
\end{theorem}

\begin{proof}
In the finite-dimensional case, the image of the imbedding has always open (Zariski)-dense image, cf.\ \cite{Be00}. Moreover, for the usual (positive)
 involution,
the unitary group is {\em compact}, hence the imbedding has also closed image. Since the Lagrangian variety is connected, 
the open and closed image must be all of it. 
\end{proof}

This compacteness argument does of course not 
carry over to the case of infinite dimensional algebras (which we need in quantum theory).
In general, the image will be open, or dense, if and only if the group of units of $\bA$ is open, resp.\ dense in $\bA$.
For instance, when $\bA$ is a Banach algebra, it is open, but need not be dense.
Thus $\UU = \cR \cap U_{NS}$ will be always open in $\cR$, but
equality will be a finite-dimensional feature. 
The four elements $O,W,F,B$ belong to $\UU$, and so do the ``linear parts'' determined by them (Theorem \ref{th:U-complete}).
One may ask if the elements of $\cR \setminus \UU$ are ``unphysical'', or if they have a ``physical meaning''. For the time being, I have no answer
to this question. Let's say that they could be considered as ``members of the multiversum who are not really admitted in our universum''...


\subsubsection{Action of the octahedron group}\label{ssec:six-poles}
Fix $(O,W,F)$ as in (\ref{eqn:OWF}) and complete them to a $6$-tuple of poles
$(O,W;N,S;F,B)$, as described above.
This $6$-tuple determines a symmetry group acting transitively on vertices, and compatible with the partition in
$3$ pairs:
 a group of 48 elements, known as the {\em octahedron group} (cf.\ Appendix \ref{app:octahedral}).
This group contains, among others, several {\em  real forms} of $\cS$, and the famous {\em Cayley transform}
permuting such real forms:

\begin{theorem}\label{th:octahedral}
We fix $6$ poles on $\cS$, with notation as above.
\begin{enumerate}
\item
The three maps $i_{a,b}$ (multiplication by the scalar $i$ with respect to the transversal pair $(a,b)$),
for $(a,b)= (N,S), (O,W), (F,B)$, generate a group $\bV_0$ of holomorphic automorphisms  of $\cS$, isomorphic to
the group 
 $\fS_4$ of direct octahedron symmetries.
  \item
The orthocomplementation map $\zeta:\cS\to \cS$, $x \mapsto x^\perp$    with respect to the 
``Euclidean (positive) Hermitian form'' on $\bA^2$ given by 
$$
\langle (u_1,u_2),(v_1,v_2) \rangle := u_1^* v_1 + u_2^* v_2
$$
is an antiholomorphic  antiautomorphism of $\cS$ exchanging opposite poles.
Composition of $\zeta$ with the maps from item (1) defines 24 antiholomorphic antiautomorphisms of $\cS$ which
permute the $6$ poles.
\item
Together, the  $48$ holomorphic and anti-holomorphic maps from (1) and (2) form a group isomorphic to 
 the full octahedron group
$\bV = \bV_0 \cup \zeta \, \bV_0 $. Its central element is $\zeta$. 
\end{enumerate}
A full description of the
group $\bV$ is given in the tables in Section \ref{app:table}.
\end{theorem}

\begin{proof}
A detailed proof is given in Section \ref{app:proof}. 
There we also describe the holomorphic maps in several ways:
by $2\times 2$-matrices with coefficients in $\bA$, acting on the complex
algebra $\bA$ by fractional linear maps in the usual way,
\begin{equation}
\begin{pmatrix} a & b \\ c & d \end{pmatrix}. z = (az+b)(cz+d)^{-1} ,
\end{equation}
as well as by ``intrinsic formulae'' realizing them as compositions of maps $\lambda_{ab}$. 
The antiholomorpic maps are also defined by matrices: 
let $M \in \Gl(2;\bA)$ an invertible $(2\times 2)$-matrix with coefficients in $\bA$.
Each such matrix defines a sequilinear form
$$
\langle (u,v),(u',v') \rangle_M  := \langle (u,v) M,(u',v') \rangle = \sum_{ij} u_i^* m_{ij} v_j .
$$
Then
the orthocomplementation map $\cS \to \cS$, $x \mapsto x^{\perp,M}$ is an antiholomorphic bijection of $\cS$ (and an anti-automorphism in the sense
of associative geometries, see \cite{BeKi2}), and the same
matrices used to describe the holomorphic bijections then also describe the antiholomorphic bijections belonging to the
octahedron group.
\end{proof}

\begin{remark}
The octahedral symmetry appears also on a more profound level as the symmetry of the whole theory of
associative lines -- see \cite{Be12}, Section 9.
There should be a link with the octahedral symmetry that we have described here, but for the moment this remains
rather mysterious.
\end{remark}

\subsubsection{Transitivity on poles, and Cayley transforms}
The group $\bV_0$ acts transitively on vertices of the octahedron (the six poles).
Eeach stabiliser group has therefore $\frac{24}{6} = 4$ elements.
The stabiliser of the north pole $N$ is the group
\begin{equation}\label{eqn:NS-stabiliser}
(\bV_0)_N = \{ \id , i_{N,S} , (-1)_{N,S}, (-i)_{N,S} \} \cong \Z / 4\Z .
\end{equation}
Each of its elements stabilises also the element $S$. 
Likewise, the stabiliser of $N$ in $\bV$ has $\frac{48}{6}=8$ elements, and hence there are also $4$ antiholomorphic elements stabilizing $N$.
By transitivity, it follows that
for each pair $(a,b)$ of vertices, there are exactly $4$ holomorphic and $4$ antiholomorphic elements $g \in \bV$ such that $g(a)=b$.
By definition of $\bV$, they have the property that then also $g(a')=b'$, when $(a,a')$ and $(b,b')$ are opposite poles.
For instance, the $4$ holomorphic transformations $g$  sending $(N,S)$ to $(W,O)$   are (notation as in Tables
\ref{ssec:table-holomorphic})
\begin{equation}\label{eqn:Cayley2}
C= (SBO)(WNF),  \quad  (NBW)(SFO), \quad (NW)(SO)(FB), \quad (OSWN) .
\end{equation}
The first of these is the (usual) {\em Cayley transform}, having order $3$; the 
second its negative $-C$; the third a transposition, and the last a $4$-cycle. 
In the same way there are four elements sending  $(N,S)$ to $(O,W)$, among them two elements of order 3
(Cayley transforms), one of order 2, and one of order 4. 

\subsubsection{Center, commutant, matrix realization}
From the definition of the $\bA$-unitary groups, and the Table in Subsection \ref{ssec:table-antiholomorphic},
we get the following descriptions of unitary groups of $2\times2$-matrices as group of maps commuting with 
antiholomorphic maps:
\begin{align*}
\Gl(2,\bA)^\zeta & = \UU(2,\bA), \\
\Gl(2,\bA)^{\zeta \circ (-1)_{O,W}} & = \UU(1,1;\bA) := \UU(I_{1,1},\bA) , \\
\Gl(2,\bA)^{\zeta \circ (-1)_{F,B}} & = \UU(F,\bA) = R \, \UU(1,1;\bA) \,R^{-1}  , \\
\Gl(2,\bA)^{\zeta \circ (-1)_{N,S}} & = \UU(J;\bA) = C \,  \UU(I_{1,1},\bA) \, C^{-1} 
\end{align*}
where the last three groups are isomorphic among each other, by conjugation via the matrices $R$, resp.\ $C$.
Intersecting the first and second of these groups, we get
$$
\Gl(2,\bA)^\zeta \cap \Gl(2,\bA)^{\zeta \circ (-1)_{O,W}}  =  \UU(\bA) \times \UU(\bA)
$$
(diagonal matrices with both diagonal entries from $\UU(\bA)$, and conjugating with the Cayley transform $C$
$$
\Gl(2,\bA)^\zeta \cap \Gl(2,\bA)^{\zeta \circ (-1)_{N,S}} = C (\UU(\bA) \times \UU(\bA) ) C^{-1}.
$$
The right hand side can be identified with $\UU \times \UU$, whence
\begin{equation}
\UU \times \UU = \{ h \in \Gl(2,\bA) \mid \, \zeta \circ h =h \circ \zeta, \, 
(-1)_{N,S}\circ h = h \circ (-1)_{N,S} \}.
\end{equation}
In words, left and right translations by elements of $\UU$ on $\cS$ are exactly those transformations commuting
with the {\em antipode map} $\zeta: z \mapsto -\overline z^{-1}$, and with the {\em central inversion map}
$z \mapsto - z^{-1}$; or with the principal involution
$\tau(z) =\overline z$, and the central inversion.
Among the elements of the octahedron group, also the elements $i_{N,S}$ and $(-i)_{N,S}$ commute with the
$\UU \times \UU$-action: the commutant of this action in the octahedron group is precisely the stabiliser group
(\ref{eqn:NS-stabiliser}). These belong to the {\em center of $\UU$}.
Of course, there are also $4$ anti-unitary transformations commuting with the $\UU\times \UU$-action; 
it seems that they have deserved only little attention so far, in quantum theory.

\begin{remark}
All the preceding results are algebraic in nature, and are valid in geometries defined over general base fields and rings.
 This is in keeping with the Jordan algebraic approach to the Cayley transform, developed by
Loos (Section 10 in \cite{Lo77}); it should be compared with the Lie theoretic approach (Koranyi-Wolf \cite{KW65})), which uses
analytic, transcendental methods (one needs the exponential map, hence some completeness assumptions on the base field, that is,
one works over the reals and mainly in finite dimension).
\end{remark}

\subsection{The unitary group contains all affine parts}
The following is the most important structural result on the unitary group, from the viewpoint of completed quantum mechanics:  it says that
$\UU$ is a ``completion'' of linear quantum mechanics, in the same way as $\cR$ has been considered its completion in Part I:

\begin{theorem}[Affine completeness of $\UU$]\label{th:U-complete}
Assume $\bA$ is a $P^*$-algebra. Then, with notation as in Definition \ref{def:U-setting}, 
the unitary torsor $\UU$ contains all affine cells defined by all of its elements: for all $a \in \UU$, 
$$
(U_a \cap \cR) \subset \UU.
$$
\end{theorem}

\begin{proof}
More formally, taking account of the definition of $\UU$,
 the claim reads
 $$
 \forall a, x \in \cS : \qquad
\bigl(\tau(a)=a, \,
\tau(x)=x, \,  a \top N  , \, x \top a  \quad \Rightarrow \quad x \top N \bigr) .
$$
Since $\UU$ acts transitively on itself, we may assume without loss of generality that here $a=W$, so $U_a = \bA$ is the usual imbedding of
$\bA$ into the complex projective line $\cS = \bA \PP^1$.
Consider the Cayley transform $C(z)=(z-i)(z+i)^{-1}$ (third line of the last table in Subsection \ref{ssec:table-holomorphic}).
Positivity is used in the following 

\begin{lemma}\label{la:Cayley}
 When $\bA$ is a $P^*$-algebra, then the four Cayley transforms are defined on all of $\Herm(\bA)$.
 In other terms, for all $x \in \Herm(\bA)$, the value $C(x)$ belongs to $\bA = U_{O,W}$.
 \end{lemma}
 
 \begin{proof}
We have to show that, for all $z \in\Herm(\bA)$, the element $z + i $ is invertible in $\bA$.
Now, $(z+i)(z-i)=z^2 - i^2 = z^* z +1$  is invertible in $\bA$ by the axioms of a $P^*$-algebra, and thus
both $z+i$ and $z-i$ are invertible, too. 
\end{proof}

To finish the proof of the theorem,
assume $z=x$ with $\tau(x)=x$. By the lemma, the value $C(x)$ is finite, which means that $C(x) \in \bA$, that is, $C(x) \top W$.
But, according to the table from Subsection \ref{ssec:table-holomorphic}, $C$ represents the permutation 
$(SBO)(WNF)$ of the six poles. Since $C^{-1}$ is an automorphism, it preserves transversality, and hence
$C(x) \top W$ implies that $x \top C^{-1}(W) = N$, which had to be shown.
\end{proof}

\begin{remark}
For the sake of the proof, one could have worked with other transformations instead of $C$: all of the holomorphic
transformations from the octahedron group sending $(N,S) \mapsto (W,O)$ (see Equation (\ref{eqn:Cayley2})) are defined everywhere on
$\Herm(\bA)$ and could be used.
 However, the Cayley transforms are  preferred,
since they belong to the group $A_4$, whereas the transpositions and $4$-cycles  belong to ${\mathfrak S}_4 \setminus A_4$.
\end{remark}

\subsection{Antipode map, and self-duality again}
In Part I, we have stressed the aspect that the projective line over an algebra is {\em self-dual}.
Now, in the unitary setting, this self-duality appears in another shape:
fixing the pair of poles $(N,S)$ as ``canonical'', the {\em antipode map} is also ``canonical'':
\begin{equation}\label{eqn:ant}
\ant : \UU \to \UU, \quad x \mapsto \ant(x) = (-1)_{N,S}(x) .
\end{equation}
The point $\ant(x)$ can be thought of as a  ``double of $x$''.
For each $p \in \UU$, the pair
$(p,\ant(p))$ is transversal, hence $U_{\ant(p)}^\tau$ is an affine space,  containing the point $p$.
Choosing $p$ as origin, this gives a vector space, denoted by $V_p : = (U_{\ant(p)},p)$, isomorphic to $\Herm(\bA)$.
In Part I, we have defined a {\em complete obstate} to be a quadruple
$(A,W;A_0,W_\infty)$ where the reference part $(A_0,W_\infty)$ is a transversal pair.
In the unitary setting (Def.\ \ref{def:U-setting}), in order to get an $\UU(\bA,*)$-invariant theory, one has to demand that both elements of the
reference part are antipodes of each other:

\begin{definition}\label{def:U-setting-bis}
A {\em complete observable}, in the unitary setting, is a pair $H =(h;p) \in \UU^2$ with $h \top \ant(p)$; 
a {\em complete state} is a pair $W = (w,p)$ with $w \top p$, and
a {\em complete obstate} is a triple $(h,w;p)$ with $w \top p, h \top \ant(p)$. 
The {\em space of complete observables, in the unitary setting}, is denoted by
$$
\cO_\UU = \{ (h,p) \in \UU^2 \mid \, h \top \ant(p)   \} .
$$
\end{definition}

\section{Time evolution in completed quantum theory}\label{sec:evolution}

\subsection{Tangent spaces: quantum convention}\label{sec:q-calculus} 
So far, in Part I and Part II, we have not yet used {\em differential calculus}.
We will start to use it now, and it is in this context that Planck's constant $\hbar$ will appear. 
As I understand the setting, on purely mathematical grounds Planck's constant is a quantity that distinguishes ``space'' from ``tangent space''.
I assume the reader is familiar with the notion of {\em tangent space} $T_p M$ of a manifold $M$ at the point $p$.
For infinite dimensional manifolds (like our $\cS, \cR$ and $\UU$), the definition of tangent spaces follows the classical pattern known in physics,
via the transformation properties of tangent vectors (see, e.g., \cite{Be08}, I.3): there is no particular problem about this.\footnote{ 
Problems arise only if you use definitions invoking, in way or another, duality of vector spaces, e.g., if you define tangent vectors as 
 point derivations. See \cite{Be08} for such issues.} 
Then, one notes that, if $M=V$ is a {\em vector space}, there is a {\em canonical identification} between $V$ and each of its tangent spaces $T_p V$.
However, if  you try, from a conceptual calculus viewpoint, to analyze what makes this identification ``canonical'', you realize that it is 
somewhat less canonical than one usually thinks.
First of all, the {\em sign} of this identification depends on your philosophy, because for $v \in V$, the differerential operator induced by the one-parameter
group $(x \mapsto x+tv)_{t \in \R}$ is the {\em negative} of the ``constant vector field $v$''.  Indeed, the sign of the ``canonical''  identification is a
{\em convention}.
But, moreover, the whole theory will not change its shape if, by convention, you plug in another, ``global'', invertible factor $\hbar$ into the ``identification
between $V$ and $T_p V$'':

\msk \nin
{\bf Quantum Convention.}
{\em
There is an invertible real number $\hbar$ such that, for all real or complex vector spaces $V$ and all $p \in V$, the map
$$
Q_p:
T_p V \to V, \quad v \mapsto \hbar v 
$$
is the correct {\em quantum identification between $V$ and the tangent space $T_p V$}. 
  } 
 
\msk
We could, in principle, forget the ``usual'' identification, and work only with the new, ``quantum'' one: since scalars commute with all linear maps, this
would not change in any essential way the shape of differential geometry. That would be the mathematical analog of choosing Planck units in physics, normalizing
$\hbar = 1$.  
For better readability of formulae, we shall use this normalization and suppress the symbol $Q_p$; but
 for sake of ``dimensional analysis'',  one should keep in mind that, whenever one identifies ``space'' and ``tangent space'' for vector spaces, then a $\hbar$-factor
 would come in.

\subsection{The tangent bundle of $\UU$}
The spaces $M = \cS, \cR, \UU$ from our setting of completed quantum mechanics are all (infinite dimensional) differentiable manifolds.
The smooth atlas is simply given by all ``affine parts'' $\bA \subset \cS$, resp.\ $\Herm(\bA) \subset \cR$, and
$\Herm(\bA) \subset \UU$ (see \cite{BeNe} for details). 
To the extent that these affine parts are canonical, their identification with tangent spaces will also be canonical.
For $p \in \UU$, recall from (\ref{eqn:ant}) the antipode $\ant(p)$ and the vector space
$V_p = (U^\tau_{\ant(p)},p)$. We shall identify it with the tangent space $T_p \UU$, with zero vector the origin $p$:
 the following linear isomorphism can be considered ``quantum canonical'':
\begin{equation}
T_p \UU \to U^\tau_{\ant(p)},\quad v \mapsto  v   
\end{equation}
(to be quantum-correct, one should write $v \mapsto Q_p(v)$ here...).
Putting all isomorphisms $T_p \UU \cong U_{\ant(p)}$ together, we 
get the quantum identification of the {\em tangent bundle} $T\UU$ with an open subset of $\UU \times \UU$: the map
\begin{equation}\label{eqn:tangent}
T\UU \to \{ (a,b) \in \UU^2 \mid a \top \infty(b) \}, \quad v \mapsto \bigl( \pi(v), v \bigr)
\end{equation}
is bijective,  where as usual $\pi:TM \to M$ is the base projection, associating to a tangent vector $v \in T_p M$ the footpoint $p \in M$. 
Comparing with Definition \ref{def:U-setting-bis}, and summarizing:

\begin{theorem}
In the unitary setting, and keeping account of the Quantum Convention, the tangent bundle $T\UU$ is identified with the space $\cO$
of complete observables, via (\ref{eqn:tangent}).
\end{theorem}

\subsection{Cotangent spaces and cotangent bundle: duality again}
Usually, the cotangent space $T_p^* M$ is defined to be some topological dual space of the (topological) vector space $T_p M$. 
However, we do not wish to enter here into technical discussions about topologies and topological duals. In the situation of completed quantum theory,
such problems have already been mentioned in the context of {\em traces}  in Part I, Appendix E.  Keeping in mind the caveats discussed there,
we {\em define the quantum cotangent space} to be vector space
\begin{equation}\label{eqn:cotangentspace}
T_p^* \UU  := U_p^\tau, \mbox{ with origin } \ant(p) ,
\end{equation}
and with the bilinear pairing between $T_p\UU = U_{\infty(p)}$ and $T_p^*\UU = U_p$ defined by  the {\em expectation value}, as explained in Part I:
for  $\phi  \in T_p^* \UU$ and $v \in T_p \UU$ let 
\begin{equation}\label{eqn:cotangent}
\langle v,\phi \rangle = \tr \bigl(\CR (p,\ant(p); v, \phi ) \bigr)  .
\end{equation}
As said in Part I,  the pairing may take infinite values for certain pairs $(v,\phi)$,  depending on the topological setting (working with unbounded operators,
etc.); but this does not affect the preceding definition.  Assembling the spaces  (\ref{eqn:cotangentspace}), we get the {\em cotangent bundle}: the map
\begin{equation}
T^* \UU \to 
\{ (a,b) \in \UU^2 \mid a \top \ant(b) \}, \quad \phi \mapsto \bigl( \pi(\phi), \ant( \phi) \bigr)
\end{equation}
is bijective.
 Note the difference with (\ref{eqn:tangent}): just one symbol $\infty$. That is, 
\begin{equation}
T_p^* \UU = U_p^\tau = U_{\ant(\ant(p))}^\tau = T_{\ant(p)} \UU
\end{equation}
with zero vector $\infty(p)$, resp.\ $p$. Thus  $T\UU = T^* \UU$ as sets, but with projections and zero sections given by
\begin{equation}
\begin{matrix}
T\UU & = & T^* \UU \\
\downarrow & & \downarrow \\
\UU & \overset{\ant}{\longrightarrow} & \UU .
\end{matrix}
\end{equation}

\begin{theorem}
In the unitary setting, and keeping account of the Quantum Convention, the cotangent bundle $T^*\UU$ is identified  with the space
of complete states, via (\ref{eqn:cotangentspace}). Using the antipode map, the spaces of complete observables and and of complete states are in bijection with each
other.
\end{theorem}


Having defined tangent and cotangent bundles, we can speak of {\em vector fields} (sections of $T\UU$) and
{\em $1$-forms} (sections of $T^* \UU$).

\subsection{The Lie algebra of the unitary group}\label{sec:LieU}
Now we shall use the fact that $\UU$ is a {\em torsor}, that is, a (Lie) group, after having fixed some base point $p \in \UU$.

\begin{definition} The {\em Lie algebra} of $\UU$ is the space $\uu$ of left-invariant vector fields on $\UU$.
\end{definition}

As usual in Lie theory, every tangent space $T_p \UU$
 can be identified with $\uu$: for each $p \in \UU$, there are inverse bijections,
the first given by
evaluation at $p$, the second  given by transporting a tangent vector by left translations to any
other tangent space,
\begin{equation}
\begin{matrix}
\uu \to T_p \UU, & \quad \xi \mapsto \xi(p)  
\\
T_p \UU \to \uu, & \quad v \mapsto \xi_v ,
\end{matrix}
\end{equation}
 where
$
\xi_v(u) = T_p \rho_{u,p} (v)
$,
and $\rho_{u,p} (x) = x p^{-1} u$ is right translation from $p$ to $u$. 
In other words, we have a diffeomorphism ``left trivialization of the tangent bundle''
\begin{equation}
\UU \times \uu \to T\UU, \quad (u, \xi) \mapsto \xi_u .
\end{equation}

\begin{theorem}
The left invariant vector fields form a Lie algebra (closed under the Lie bracket of vector fields).
\end{theorem}

\begin{proof}
For the general proof (in arbitrary dimension), see eg., \cite{Be08}, I.5.3. 
In our special case it can moreover be shown that $\uu$ is a subalgebra of the {\em conformal Lie algebra} (Lie algebra of $\PP \Gl(2,\bA)$, in our
case), and hence its elements extend  to vector fields that are defined on all of $\cS$, see \cite{Be00, BeNe}.
\end{proof}

\nin
Since the choice of base point in $\UU$ is arbitrary,  all tangent spaces $T_p\UU$ are Lie algebras, and we have in fact
defined a {\em field of Lie algebras} on $\UU$, and
since $T_p \UU$ can also be identified with the Jordan algebra $\Herm(\bA)$, it is in fact a Jordan-Lie algebra.
The Lie algebra structure reflects  ``infinitesimal'' aspects of our setting, whereas the Jordan algebra structure rather reflects ``global''
aspects; therefore by our Quantum Convention, a factor $\hbar$ comes in, which corresponds to the Jordan-Lie
constant $\kk$ (cf.\ Remark \ref{rk:hbar}).

\subsection{Flow equation in completed quantum mechanics}
The preceding isomorphisms can be written
\begin{equation}
\begin{matrix}
\cO_\UU &  \cong & T\UU & \cong &  \UU \times \uu,
\end{matrix}
\end{equation}
and thus a complete observable $H = (h;p) \in \cO_\UU$ corresponds to a tangent vector $v \in T\UU$, and to a pair $(p,\xi_H)$ with a left invariant 
vector field $\xi_H\in \uu$. 

\begin{definition}
The {\em geometric Schr\"odinger equation} for the Hamiltonian $H = (h;p) \in \cO_\UU$
 is the flow equation of the left invariant vector field $\xi_H$:
 the flow $\Psi:\R \times \UU \to \UU$ is solution of
\begin{equation}
\frac{d}{dt} \Psi (t,x) =   \xi_H (\Psi(t,x)).
\end{equation}
\end{definition}

\subsubsection{Solution of the flow equation}
Under too weak topological assumptions on the algebra $\bA$, the flow equation need not admit a solution, nor will we have uniqueness of solutions.
However, under usual assumptions (e.g., $\bA$ is a Banach algebra), the algebra $\bA$ will admit an {\em exponential map}: 
 the usual exponential series
$e^X =\sum_{k=0}^\infty \frac{X^k}{k!}$ converges in $\bA$, and
$\frac{d}{dt} e^{tX} p = X (e^{tX}p)$. 
This implies that the Lie groups $\bA^\times$ and $\UU(\bA,*)$ also admit exponential maps, and by the general theory of Lie groups, 
the  solution of the flow equation will be given by
\begin{equation}
\Psi(t,x) = x . \exp_\UU ( t   \xi_h).
\end{equation}
Since $\UU \cong \UU(\bA,*)$, in terms of the algebra exponential, $\exp( t \xi_h)$ corresponds to $e^{it \xi_h}$.
For practical computations, one will return to the classical, linear, picture of time evolution: via the Cayley transform, 
the vector field $\xi_h$ is realized as a linear vector field on $\bA$, and we are back in ``business as usual''.
However,
transforming back, via the Cayely transform, the vector field can also be realized in ``Jordan coordinates'' (cf.\ \cite{Be00}) as a quadratic vector field.
Integrating a quadratic vector field leads to more complicated formulas (cf. \cite{Be00}, Section X.4).

\subsubsection{Equivalence of pictures}
The general pattern of Lie theory shows that each of the following determines the other, for a Lie group $G$:
\begin{enumerate}
\item
the adjoint representation $\Ad: G \to \GL(\g)$
\item
the coadjoint representation $\Ad^*: G \to \GL(\g^*)$
\item
the action of $G$ on its tangent bundle $G \times TG\to TG$
\item
the left (or right) action of $G$ on itself, $G \times G \to G$. 
\end{enumerate}

\nin
Mathematically, the equivalence of our geometric picture of unitary time evolution with the Schr\"odinger picture (1) or the
Heisenberg picture (2) follows from this pattern. 
Let's recall the basic arguments:
if the origin, say $p=O$, is considered to be fixed, then one considers the action of $G$ on itself by conjugation,
$G \times G \to G$, $(g,h) \mapsto ghg^{-1}$. It fixes $O$, and we can derive at $O$ and get
$\Ad:G \times \g \to \g$ where as usual $\g = T_O G$.
Likewise, we get $\Ad^*:G \times \g^* \to \g^*$.
The bilinear pairing $\g \times \g^* \to \R$ is $G$-invariant, that is,
\begin{equation}\label{eqn:Ad}
\langle \Ad(g) v, \Ad^*(g) \phi \rangle = \langle v,\phi \rangle,
\end{equation}
whence $\langle \Ad(g) v,  \phi \rangle = \langle v,\Ad^*(g)^{-1}\phi \rangle$,
which for $g = \exp(t \xi)$ gives 
\begin{equation}\label{eqn:Ad2}
\langle \Ad(\exp(t\xi)) v, \phi \rangle = \langle v ,\Ad^*(\exp( - t \xi)) \phi \rangle .
\end{equation}
The left hand side term describes time evolution of expectation values in the Schr\"odinger picture, and the right
hand side in the Heisenberg picture: both are equivalent  (cf. e.g. \cite{Tak08} p. 77 for this discussion).

\ssk
Now consider $G$ as a torsor, that is, no point of $G$ plays a distinguished role: ``all points are created equal''.
For $x,y \in G$, there is a {\em left translation from $y$ to $x$},
$L_{x,y}(z)=xy^{-1}z$, and a {\em right translation from $y$ to $x$},
$R_{x,y}(z)=zy^{-1}x$. When $y=O$, these are just the usual left and right translations by $x$.
The canonical pairing between $TG$ and $T^*G$ is invariant under 
general diffeomorphisms, hence is invariant {\em both} under the left and right action:
\begin{equation}\label{eqn:Ad3}
\forall p\in G, \forall v \in T_p G, \forall \phi \in T_p^* G: \quad
\langle v,\phi \rangle_p = \langle v.g , \phi.g \rangle_{p.g}  = \langle g.v,g.\phi \rangle_{g.p}.
\end{equation}
Thus, 
if we let act $G = \UU(\bA)$ on complete obstates from the left, or from the right, in the ``obvious'' way, and let
$g = \exp(t \xi_h)$,
then expectation values do not evolve at all: they are constant. 
If we want to ``observe expectation values that evolve'', we have to rewrite this condition in a similar way as
(\ref{eqn:Ad}) has been transformed into (\ref{eqn:Ad2}).
To this end, we let act $G$  on $TG$ via the left translations, and  on $T^* G$ via the right translations:
given a tangent vector $\xi \in T_p G$, let $q := \exp_p(\xi) \in G$ its image under the exponential map
defined at $p$ (note that the exponential map does not depend on choices: it is the same when working with
right or left invariant vector fields).
Then for $(v,\phi) \in T_p G \times T^*_p G$, let
\begin{equation}
(v_t,\phi_t) := (T_p L_{q,p} v, T^*_p R_{q,p} \phi) \in T_q G \times T_q^* G .
\end{equation}
Of course, the letters $L$ and $R$ could also be interchanged; the important point  is that both are used. 
Then, since $\Ad(g) = L_g \circ R_g^{-1}$, by invariance of the pairing under all left and all right translations,
 the expectation value coincedes with the time-depending expectation value from the usual linear theory:
\begin{equation}
\langle v_t, \phi_t \rangle_q = \langle \Ad(\exp(t \xi)) v, \phi  \rangle_p .
\end{equation}
Summing up,  these pictures are all mathematically equivalent --
the question whether one or the other of these pictures fits better with {\em physical interpretations}, remains, of course, open. 

\section{Some concluding remarks}

\subsection{Comparison with Hamiltonian mechanics}
Many  textbooks ``motivate'' the formalism of quantum mechanics by its analogy with the one of Hamiltonian mechanics.
Indeed, there is a strong structural analogy:

\msk
\begin{center}
\begin{tabular}{|l|c|c |}
  \hline
   & classical  & complete quantum theory \\
  \hline
 Hamiltonian as observable  &   $H \in F(M,\R)$   &  $ H =(p, h)$  \\
 \hline
Hamiltonian as vector field  & $X_H$  &  $\xi_H$
  \\
 \hline
 evolution equation  & flow of $X_H$   &  flow of $\xi_{H}$  \\
  \hline
\end{tabular}
\end{center}
\msk


\nin However,
I have the impression that  the chain of motivation should rather go the other way round:
Hamiltonian mechanics  historically precedes quantum mechanics, but logically and mathematically the quantum
side should be prior. 
The crucial ingegredients in the scheme are {\em duality}, and {\em differential calculus}: 

\msk
\begin{center}
\begin{tabular}{|l|c|c |}
  \hline
   & classical  & quantum  \\
  \hline
differentiate   &   $H \mapsto dH$   &  multiply by $\frac{1}{\hbar}$  \\
 \hline
dualize   & via symplectic form, or via Poisson-tensor  &  multiply by $i$ \\
\hline
 \end{tabular}
\end{center}
\msk

\nin
I think that both 
 topics, duality and differential calculus, could be better understood in the geometric, ``completed'', approach,
 and that the link with their classical roles should become clearer.

\subsection{On the measurement problem (Part III ?)}
We have not touched, neither in Part I nor in the present Part II, on the \href{https://en.wikipedia.org/wiki/Measurement_problem}{``Measurement problem''}.
Since our geometric setting is, mathematically, equivalent to the common linear setting, all ``solutions'' and ``interpretations'' that have been proposed,
could in principle be transferred to the geometric setting. For the time being, I have not succeeded in understanding the geometric and conceptual structures
that are relevant in this context. I expect that one should once again modify the setting:
the ``unitary setting'' introduced in this part is likely not to be the last word. 


\appendix

\section{More on Jordan-Lie algebras}\label{app:JLA}

\subsection{Tensor products}
For general Jordan algebras, there is no such thing as a ``tensor product of Jordan algebras'', nor is there for general
Lie algebras. Remarkably, for Jordan-Lie algebras the situation is better:

\begin{theorem}
Assume $V,W$ are Jordan-Lie algebras with same (non-zero) Jordan-Lie constant $\kk$. 
Then the $\K$-module
$V \otimes_\K W$ carries a natural structure of Jordan-Lie algebra with  Jordan-Lie constant $\kk$.
Its Jordan and Lie products are given by
\begin{align*}
(a \otimes b) \bullet (a' \otimes b') & = 
 (a \bullet a' ) \otimes (b \bullet b')  -\kk  [a,a']\otimes [b,b'] ,
\\
[(a\otimes b),(a'\otimes b')] & =
 (a\bullet a')\otimes [b,b'] + [a,a']\otimes (b\bullet b' )  .
\end{align*}
\end{theorem}

\begin{proof}
It is possible, though somewhat lengthy, to check directly the defining properties (JL1) -- (JL4).
A quicker proof is given by using Theorem \ref{th:JL1}. For simplicity,
assume first that $-\kk = \frac{w^2}{u^2}$ is a square in $\K$. 
Then $V$ and $W$ carry structures of associative algebras, inducing the Jordan-Lie structure according to
Theorem \ref{th:JL1}.
Let $V \otimes W$ be the 
\href{https://en.wikipedia.org/wiki/Tensor_product_of_algebras}{tensor product of the associative algebras} $V$ and $W$.
This is an associative algebra. We decompose the associative product into its symmetric and skew-symmetric parts:
\begin{align*}
(a \otimes b) \cdot (a' \otimes b') & = (aa') \otimes (bb') \\
& = (\frac{1}{2w} a \bullet a' + \frac{1}{2u} [a,a']) \otimes 
 (\frac{1}{2w} b \bullet b' + \frac{1}{2u} [b,b']) 
 \\
&= \frac{1}{4w^2} (a \bullet a' ) \otimes (b \bullet b') + \frac{1}{4 u^2} [a,a']\otimes [b,b'] +
\\
& \qquad \frac{1}{4 uw} ( (a\bullet a')\otimes [b,b'] + [a,a']\otimes (b\bullet b' ) ) 
\\ 
&= \frac{1}{2w} \Bigr(   \frac{1}{2w} (a \bullet a' ) \otimes (b \bullet b') + \frac{-k}{2w} [a,a']\otimes [b,b']     \Bigl) + \\
 & \qquad \frac{1}{2u} \Bigr(  \frac{1}{2w} ( (a\bullet a')\otimes [b,b'] + [a,a']\otimes (b\bullet b' ) )   \Bigl)
\end{align*}
The first term is a symmetric product and the second skew-symmetric, whence, again by Theorem \ref{th:JL1}, the following
two products define a Jordan-Lie algebra structure on $V \otimes W$
\begin{align*}
(a \otimes b) \bullet (a' \otimes b') & = 
\frac{1}{2w} (a \bullet a' ) \otimes (b \bullet b') + \frac{-k}{2w} [a,a']\otimes [b,b'] ,
\\
[(a\otimes b),(a'\otimes b')] & =
\frac{1}{2 w} ( (a\bullet a')\otimes [b,b'] + [a,a']\otimes (b\bullet b' ) ) .
\end{align*}
Now we choose $w=\frac{1}{2}$, $k=- \frac{1}{u^2}$, giving the formulae from the claim.
(One realizes that the choice of $w$ gives a degree of freedom in defining the tensor product of 
Jordan-Lie algebras.) 
If $-\kk$ is not a square in $\K$, then we work in the scalar extension of $V$ and $W$ by the ring
$R=\K[X]/(X^2 + \kk)$: we take the associative tensor product
$$
V_R \otimes_\K V_R = (\Herm(V_R) \oplus j \Aherm(V_R) ) \otimes_\K (\Herm(V_R) \oplus j \Aherm(W_R)) ,
$$
and
$\Herm(V_R) \otimes_\K \Herm(W_R)$ is indeed stable under the products $\bullet$ and $[-,-]$
defined above.
\end{proof}

\begin{remark}
The problem of defining tensor products of algebras  is the starting point of the paper \cite{GP76},
where the notion of {\em composition class} as a class of two-product
algebras closed under tensor products is introduced;  essentially, Jordan-Lie algebras are solution of the problem.
The idea to characterize quantum and classical mechanics as certain
composition classes  goes back to Niels Bohr.
\end{remark}

\subsection{Relation with the Jordan-Lie functor}\label{app:LJ}
There is ``ternary version of Jordan-Lie algebras'', sometimes called {\em Lie-Jordan algebras}, cf.\ references given in \cite{Be08b}.
Every (binary) Jordan algebra $(V,\bullet)$ gives rise to a ternary product, the {\em Jordan triple system (JTS)}
\begin{equation}
T(x,y,z) = (x\bullet y) \bullet z + x \bullet (y \bullet z) - y \bullet (x \bullet z) .
\end{equation}
For instance,
when $a \bullet b = w(ab+ba)$ in an associative algebra with product $ab$, then
$T(x,y,z) = 2 w^2 (xyz+zyx)$.
On the other hand, every JTS gives rise to a {\em Lie triple system (LTS)} $R=R_T$ via
\begin{equation}
R_T(x,y,z) = T(y,x,z) - T(x,y,z) .
\end{equation}
In \cite{Be00},
the correspondence $T\mapsto R_T$ has been called the {\em Jordan-Lie functor}.  Geometrically,
$R_T$ is the curvature tensor of a symmetric space that can be associated to $T$.
(However, the sign of $R_T$ is  a matter of quite delicate conventions.)
In the example of an associative algebra with product $xy$ and
 $T(x,y,z) = 2 w^2 (xyz+zyx)$, let $[x,y]=u(xy-yx)$;  then, by direct computation, 
$$
R_T(x,y,z) = 2 w^2 (yxz-xyz + zxy - zyx) = 2 \frac{w^2}{u^2} [[y,x],z] .
$$
In a Jordan-Lie algebra, condition (JL4) therefore gives
 $R_T(x,y,z) = - 2 \kk [[y,x],z]$.
When $\kk=-1$, this means that the curvature of the symmetric space is the triple Lie bracket of the Lie algebra,
hence the symmetric space is a Lie group, considered as symmetric space.
Indeed, the symmetric space of an ordinary associative algebra $\bA$ is the group $\bA^\times$, considered as
symmetric space.
When $\kk = + 1$, the condition means that the curvature is the negative of the triple Lie bracket of the Lie algebra,
that is, the symmetric space is the {\em $c$-dual symmetric space} of the Lie group belonging to the Lie algebra.
The symmetric space of a $*$-algebra is the ``cone'' $\bA^\times / \UU(\bA)$, which is a quotient of a
complex Lie group by a real form, hence has $c$-dual a group type space, namely $\UU(\bA)$. 
However, these concepts seem a bit too general for the setting of completed quantum mechanics: they also cover 
{\em orthogonal groups}, and not only unitary ones; that is, the complex structure (which is so important for quantum
mechanics) cannot be reconstructed from such data.

\subsection{On axiomatic definition of Jordan-Lie geometries}\label{sec:JL-geometry}
We do not attempt, in this text, to give an ``axiomatic'' definition of geometries belonging to Jordan-Lie algebras.
For Jordan-Lie constant $\kk = -1$, these are the {\em associative geometries} from \cite{BeKi};
however, it is not at all clear how to adapt axiomatics to the case of Jordan-Lie constant $\kk = 1$. 
In the present text, the corresponding object is defined 
 ``by construction'': the {\em unitary setting} from Definition \ref{def:U-setting}.
However, in the long run, it should be important to understand this geometry from an axiomatic and conceptual point of view.

\section{Action of the octahedral group}\label{app:octahedral}

\subsection{The abstract octahedral group}
By definition, the \href{https://en.wikipedia.org/wiki/Octahedral_symmetry}{\em (abstract) octahedral group} is the symmetry group of the regular octahedron,
which is the same as the symmetry group of the usual $3$-cube. It is a semidirect product of ${\mathfrak S}_4$ with $\Z/2 \Z$ and thus has 48 elements.
Recall its abstract definition:
we define the following sets of three, resp.\ six elements
\begin{equation}
S_0:= \{ 1,2,3\}, \qquad S:= \{ 1,2,3, 1',2',3' \} .
\end{equation}
We say that $S = \{ 1,1' \} \cup \{ 2,2'\} \cup \{3,3'\}$ is the {\em canonical partition}, or {\em canonical equivalence relation} $\sim$,
on $S$.

\begin{definition}
The {\em octahedral group $\bV$} is the subgroup of permutations $\sigma$ in $\fS_6 = \Bij(S,S)$ that are compatible with the canonical 
equivalence relation. In other words, whenever $j = \sigma(i)$, then also $j' = \sigma(i')$, or:
$\forall i$,  $\sigma (i') = (\sigma(i))'$.
\end{definition}

\nin
Directly from the last condition, we see that the permutation $\zeta$ defined by
 $\zeta(i)=i'$ belongs to the center of $\bV$, and hence we get a canonical
morphism 
\begin{equation}
\phi: \bV \to \fS_3 = \Bij(S/\sim), \quad \sigma \mapsto [\sigma] .
\end{equation} 
This morphism splits, by letting act $\fS_3$ on $S$ by permuting the $3$ symbols $1,2,3$,
 and its kernel
 is $(\fS_2)^3$, acting in each equivalence class by transposition or 
identity. This exhibits the structure of $\bV$ as a semidirect product
\begin{equation}\label{eqn:V1}
\bV \cong \fS_3 \rtimes (\fS_2)^3 ,
\end{equation}
whence $\vert \bV \vert = 48$. 
This presentation describes the action of $\bV$ on the six vertices of the octahedron (the set $S$).
On the  other hand, $\bV$ acts also on the $8$ faces of the octahedron (corresponding to the $8$ vertices of the cube), as follows:
the canonical partition of $\cS$ has 
$8$ sets of representatives (one for each subset of $S_0$), giving rise to $4$ equivalence relations on $S$ that are
transversal to $\sim$ (meaning that each equivalence class is a system of representatives for $\sim$).
By letting act $\bV$ on subsets of $S_0$, we get an injective morphism
$\bV \to \fS_8$, and a morphism 
$\bV \to \fS_4$ having kernel $\{ \id,\zeta \}$.
 For reasons of cardinality, the latter morphism must be surjective, exhibiting the structure of $\bV$ as
 \begin{equation}\label{eqn:V2} 
 \bV = \bV_0 \times \{ \id ,\zeta \} \cong  \fS_4 \times \fS_2 
 \end{equation}
 where $\bV_0$ is a subgroup of $\bV$ isomorphic to $\fS_4$. 
Comparing (\ref{eqn:V1}) with (\ref{eqn:V2}), we realize $\fS_4$ as semidirect product of 
$\fS_3$ with the {\em Klein $4$-group}  $V = \fS_2 \times \fS_2$.

\subsection{Action of the octahedral group on the projective line $\bA\PP^1$}\label{app:table}
Next we are going to describe how the octahedral group $\bV$  acts on the projective line $\cS = \bA \PP^1$ of a $*$-algebra $\bA$.
We fix the $6$-tuple of poles 
$(O,W; F,B; N,S)$ (east-west, front-back, north-south), as in Subsection \ref{ssec:six-poles}. 
It corresponds to the $6$-tuple   $(1,1';2,2';3,3')$ from the preceding subsection.
Our aim is realize the abstract group $\bV$ as a group of holomorphic and antiholomorphic transformations of $\cS$, that is generated
by certain geometrically defined transformations (Theorem \ref{th:octahedral}), see proof given below, Subsection \ref{app:proof}. 
We start by recording the effect of certain intrinsically defined transformation that permute the $6$ poles. 
 We alsor give the matrix (linear map $\bA^2 \to \bA^2$ inducing this map), 
and finally the ``complex coordinate formula'', obtained by choosing
$(O,W,F)=(0,\infty,1)$ as base point used to identify the Riemann sphere taken out $\infty$ with $\bC$ (then
$(B,N,S)=(-1,i,-i)$).
The matrix formulae generalize to the sphere $\cS$ of an arbitary $*$-algebra $(\bA,*)$, upon replacing
$\overline z$ by $z^*$. 

\subsubsection{
List of holomorphic transformations.}\label{ssec:table-holomorphic}
 There are 24 of them; they form a group $\bV_0 \cong \fS_4$, and we organize our
list following the structure of that group: first its normal subgroup $V$,
then the six $4$-cycles,
then the six transpositions not in $V$,
and finally the eight $3$-cycles:

\msk
\nin
\begin{tabular}{l | l | l | l}
elements of Klein $4$-group   & intrinsic formula  &  matrix &  complex  formula \cr
\hline
$\id$ & $\id$ & $\bigl( \begin{smallmatrix} 1& 0\\ 0&1 \end{smallmatrix} \bigr)$ &  $z \mapsto z$
\cr 
$(NS)(FB)$
&$ (-1)_{O,W}$ & $\bigl( \begin{smallmatrix}-1 &0 \\ 0& 1 \end{smallmatrix} \bigr)$ 
& $z \mapsto -z$
\cr
$(OW)(NS)$ &
$ (-1)_{F,B} $&  $\bigl( \begin{smallmatrix}0 &1 \\ 1 &0 \end{smallmatrix} \bigr)$ &
$z \mapsto z^{-1}$
\cr
$(OW)(FB))$ &
$ (-1)_{N,S}$ & $\bigl( \begin{smallmatrix} 0& 1 \\ -1 &0 \end{smallmatrix} \bigr)$  &
$z \mapsto - z^{-1}$
\end{tabular}

\msk
\nin
Remark. 
We have $(-1)_{O,W} = J^{FB}_O = J^{FB}_W = $ and
$(-1)_{F,B} =  J^{OW}_F$, etc.; thus we get symmetries of the form $J^{xz}_y$ with $(x,z)$ a pair of opposite poles.
Note that, if $x,y$ are not opposite poles, then $J^{xy}_z$ does not preserve our octahedron. 

\msk 
\nin
\begin{tabular}{l | l | l | l}
$4$-cycles in $\fS_4$ & intrinsic formula  & matrix & 
complex formula  \cr
\hline
$(FNBS)$ & $i_{O,W}$& $\bigl( \begin{smallmatrix} i& 0\\ 0&1 \end{smallmatrix} \bigr)$ &  
$z \mapsto i z$
\cr
$(SBNF)$ & $(-i)_{O,W} = i_{W,O}$& $\bigl( \begin{smallmatrix} - i & 0\\ 0&1 \end{smallmatrix} \bigr)$ &  
$z \mapsto -i z$
\cr
\hline
$( FWBO) $& $ i_{N,S}$ &  $R = \bigl( \begin{smallmatrix} 1& -1 \\ 1 &1 \end{smallmatrix} \bigr)$ &  
$z \mapsto (z-1)(z+1)^{-1}$ 
\cr
$(OBWF)$& $ (-i)_{N,S} = i_{S,N}$& $R^* = \bigl( \begin{smallmatrix} 1& 1 \\  -1 &1 \end{smallmatrix} \bigr)$ &  
$z \mapsto - (z+1)(z-1)^{-1}$
\cr
\hline
$(NWSO)$
& $i_{F,B}$&  $\bigl( \begin{smallmatrix} 1& i \\ i &1 \end{smallmatrix} \bigr)$ &  
$z \mapsto  (z+i)(iz+1)^{-1}$
\cr
$(OSWN)$
& $(-i)_{FB}=i_{B,F}$& $\bigl( \begin{smallmatrix} 1& -i \\ - i&1 \end{smallmatrix} \bigr)$ &  
$z \mapsto   (z-i) (-iz +1)^{-1}$
\end{tabular}

\msk
\nin
In the preceding table,  horizontal lines arrange a $4$-cycle together with its inverse.
In the following table, they arrange a transposition together with a transposition commuting with it:

\msk
\nin
\begin{tabular}{l | l | l | l}
transposition in $\fS_4$  
& intrinsic formula &  matrix & complex formula \cr
\hline
$(NF)(SB)(OW)$ & $ (-1)_{F,B}\circ i_{W,O} = (-1)_{N,S} \circ i_{O,W}  $ & $\bigl( \begin{smallmatrix} 0 & 1\\ i &0 \end{smallmatrix} \bigr)$ &   $z\mapsto -i z^{-1}$ 
\cr
$(NB)(SF)(OW)$ & $ (-1)_{F,B}\circ i_{O,W} = (-1)_{N,S} \circ i_{W,O}  $
 & $\bigl( \begin{smallmatrix} 0 & i \\ 1 &0 \end{smallmatrix} \bigr)$ &   $z \mapsto i z^{-1}$ 
\cr
\hline
$(FO)(BW)(NS)$ &  $(-1)_{F,B} \circ i_{S,N} = (-1)_{O,W} \circ i_{N,S}$ & $\bigl( \begin{smallmatrix} -1& 1\\ 1&1 \end{smallmatrix} \bigr)$ &  $z\mapsto (1-z)(z+1)^{-1}$
\cr
$(FW)(BO)(NS)$ &  $(-1)_{F,B} \circ i_{N,S} = (-1)_{O,W} \circ i_{S,N}$
 & $\bigl( \begin{smallmatrix} 1 & 1\\ 1&-1 \end{smallmatrix} \bigr)$ &  $z\mapsto (z+1)(z-1)^{-1}$
\cr
\hline
$(NO)(SW)(FB)$ & $(-1)_{N,S}\circ i_{B,F} = (-1)_{O,W}\circ i_{F,B}$
  & $\bigl( \begin{smallmatrix} -i & 1 \\ -1  &i \end{smallmatrix} \bigr)$ &  $z\mapsto  (1-iz)(i-z)^{-1}$
\cr
$(NW)(SO)(FB)$ &  $(-1)_{N,S}\circ i_{F,B} = (-1)_{O,W}\circ i_{B,F}$
  & $\bigl( \begin{smallmatrix} -1& i \\ -i &1 \end{smallmatrix} \bigr)$ &  $z\mapsto (i-z)(1-iz)^{-1}$
\end{tabular}

\msk
\nin
The composition of two commuting transpositions belongs to $V$, 
and the composition of any of the other two transpositions gives a $3$-cycle. Altogether, we get eight $3$-cycles
(one for each face of the octahedron),
which all deserve to be called a ``Cayley transform''.
The \href{https://en.wikipedia.org/wiki/Cayley_transform#Complex_homography}{``official''  Cayley transform}  is
given in the  third line: 
$C.z = (z-i)(z+i)^{-1}$.  
We arrange a cycle   together with its inverse:

\msk
\nin
\begin{tabular}{l | l | l | l}
$3$-cycles in $A_4$ 
  & intrinsic formula  & matrix &
complex formula  \cr
\hline
$(NBO) (SFW) $
&$ i_{O,W} \circ i_{F,B} $ & $\bigl( \begin{smallmatrix} i & -1\\  i&1 \end{smallmatrix} \bigr)$ 
& $z \mapsto  (z+i)(z-i)^{-1}$
\cr
$(NOB) (SWF)$
&$i_{B,F} \circ  i_{W,O} $ & $\bigl( \begin{smallmatrix}  - i & -i \\ -1 &1 \end{smallmatrix} \bigr)$ 
& $z \mapsto (iz+1)(z-1)^{-1}$
\cr
\hline
$(SBO) (WNF)$ &
$ i_{N,S}\circ i_{W,O}$ & $\bigl( \begin{smallmatrix} 1 & -i \\ 1&i \end{smallmatrix} \bigr)$ &  
$z \mapsto (z-i)(z+i)^{-1}$
\cr
$(SOB)(NWF)$ & $i_{O,W}\circ i_{S,N}$
 & $\bigl( \begin{smallmatrix} i & i\\ i &-1 \end{smallmatrix} \bigr)$ &  
 $z \mapsto i(z+1)(1-z)^{-1}$
  \cr
  \hline
$(NBW)(SFO)$ & $i_{O,W}\circ i_{B,F}$
 & $\bigl( \begin{smallmatrix} -1&i \\ 1 &i  \end{smallmatrix} \bigr)$ &  
 $z \mapsto (i-z)(z+i )^{-1}$
\cr
$(WBN)(FSO)$ &  $i_{F,B} \circ i_{W,O}$
& $\bigl( \begin{smallmatrix} -i& i\\ 1&1 \end{smallmatrix} \bigr)$ &  
$z \mapsto i (1-z)(1+z)^{-1}$
\cr
\hline
$(SWB)(NOF)$ & $i_{O,W}\circ i_{N,S}$
 & $\bigl( \begin{smallmatrix} i & -i\\ 1&1 \end{smallmatrix} \bigr)$ &  
 $z\mapsto i (z-1) (z+1)^{-1}$
\cr
$(SBW)(NFO)$ & $i_{S,N}\circ i_{W,O}$
& $\bigl( \begin{smallmatrix} 1 & i \\ -1  &i \end{smallmatrix} \bigr)$ &  
$z \mapsto (z+i) (i-z)^{-1}$
\end{tabular}

\subsubsection{Antiholomorphic transformations.}\label{ssec:table-antiholomorphic}
In the standard chart, the anthiholomorphic transformation $z \mapsto z^*$ does not belong to the central element, but describes
the ``Hermitian real form'' $\tau^{NS}_O$. 
The central element is the {\em antipode map}
$\zeta = \tau^{NS}_F \circ (-1)_{N,S}$, given by the complex formula
$z \mapsto - \overline z^{-1}$.
It is the orthocomplement map with respect to the positive (``Euclidean'') form on $\bA^2$. 
On the usual Riemann sphere, it has no fixed point.
Via $f \leftrightarrow \sigma \circ \zeta$, we may again identify the 24 antiholomorphic maps with elements of $\fS_4$,
and organize the tables as above.
However, we will not give the full list, since most of them  are hardly used.
We only list the three ``major'' real forms belonging to $\zeta \circ g = g \circ \zeta$, where $g$ is a holomorphic transformation
belonging to the Klein $4$-group. They are orthocomplementation maps with respect to  forms given by the
``form matrices'' defined in equation (\ref{eqn:formmatrices}): 

\msk
\nin
\begin{tabular}{l | l |   l | l }
Klein $4$-torsor   & intrinsic formula  &  matrix & complex  formula \cr
\hline
 $(NS)(OW)(FB)$ & $\zeta$  & $1$ &  $z \mapsto - \overline z^{-1}$ antipode (central)
\cr 
$(NS)$
&$\tau^{NS}_F = \zeta \circ (-1)_{N,S}$ & $J$ 
& $z \mapsto \overline z$ Hermitian real form
\cr
$(FB)$ &
$\tau^{FB}_O = \zeta \circ (-1)_{F,B}$& $F$ & 
$z \mapsto - \overline z$ skew-Hermitian real form
\cr
$(OW)$ &
$\tau^{OW}_F = \zeta \circ (-1)_{O,W}$ & $I_{1,1}$ & 
$z \mapsto  \overline z^{-1}$ unitary real form
\end{tabular}

\msk
\nin
Each of the eight Cayley transforms permutes the three real forms cyclically, while commuting with $\zeta$. 
There are also six other,  ``minor'' or ``diagonal'', real forms $g \circ \zeta$, where $g$ is one of the six transpositions. 
For instance, 
$(FS)(NB)$ corresponds to $\zeta \circ (-1)_{F,B}\circ i_{W,O}$, given by the complex formula
$z \mapsto i \overline z$. The interested reader may write up the complete list, as well
as those of the antiholomorphic transformations of order $3$ and $4$. 

\subsubsection{Proof of Theorem \ref{th:octahedral}}\label{app:proof}
To give a most intrinsic proof, start with a transversal triple; without loss of generality, we may assume that
it is of the form $(O,W;F)$. Then define 
$B:=(-1)_{O,W} F$, $N:= i_{O,W} F$, $S:=i_{O,W} B = i_{W,O} F$.
Thus, by definition, the transformations 
$(-1)_{O,W}$, $i_{O,W}$, and $i_{W,O}$ preserve the set of $6$ vertices, and have the description given in the tables.
To compute formulae for other transformations,  we use that,
since $\Gl(2,\R)$ acts by automorphisms of the geometry, for every $g \in \Gl(2,\R)$ we have
$g \circ \lambda_{a,b}\circ g^{-1} = \lambda_{g.a,g.b}$.
In particular, if $g$ permutes the $6$ vertices, we may apply this with $\lambda \in \{ -1, i, -i\}$ and $(a,b)=(O,W)$,
to get formulae for $\lambda_{g.a,g.b}$.
Define the matrices, belonging to $\Gl(2,\bA)$,
\begin{equation}\label{eqn:formmatrices}
R:= \begin{pmatrix} 1 & -1 \\ 1 & 1 \end{pmatrix}, \, 
J:=
\begin{pmatrix} 0 & 1 \\ -1 & 0\end{pmatrix},
\,
F :=
\begin{pmatrix} 0 & 1 \\  1 & 0 \end{pmatrix},
\,
I_{1,1} :=
\begin{pmatrix} -1 & 0 \\ 0 & 1 \end{pmatrix}.
\end{equation}
The matrix $I_{1,1}$ describes $(-1)_{O,W}(z)=-z$. 
The matrix $R$  corresponds, in the usual chart, to the transformation
$R.z = (z-1)(z+1)^{-1}$, sending
$0 \mapsto -1 \mapsto \infty \mapsto 1 \mapsto 0$, that is,
$O \mapsto B \mapsto W \mapsto F \mapsto O$; if fixes $N$ and $S$.
Taking $g = R$, we get
$$
(-1)_{B,F}.z= R \circ (-1)_{O,W} \circ R^{-1} .z= (R I_{1,1} R^{-1}).z = F.z = z^{-1} ,
$$
and from this,
$(-1)_{N,S}.z = (-1)_{i F,i B}.z =i (-1)_{F,B} (i^{-1}z) = i^2 z^{-1} = -z^{-1}$,
whence the description of the Klein $4$-group given in the first table.
Similarly, 
$$
i_{B,F}.z = 
R \circ i_{O,W} \circ R^{-1}(z) =
R (i (z+1)(1-z)^{-1}) =
(z+i)(iz+1)^{-1} ,
$$
and $i_{N,S}= i_{O,W}\circ i_{B,F}\circ i_{O,W}^{-1}$, whence
$$
i_{N,S} (z) =  i \, (-iz + i) (z+1)^{-1} = (z-1)(z+1)^{-1} = R(z) ,
$$
so $i_{N,S}=R$, and we get the formulae for the $4$-cycles.
Next, to describe the transpositions, we compute
$$
((-1)_{O,W} \circ i_{N,S})^2 = \bigl( (-1)_{O,W} \circ  i_{N,S} \circ (-1)_{O,W}\bigr)  \circ i_{N,S} =
i_{S,N} \circ  i_{N,S} = \id .
$$
In the same way, whenever $(u,v)$ and $(x,y)$ are two different pairs of opposite poles,
$((-1)_{u,v}\circ i_{x,y})^2 = \id$. This gives 12 elements of order two, but by relations already established, the number
reduces to 6 (cf.\ table).
Finally, the $3$-cycles are given by compositions of two transpositions which do not commute, for instance
$$
g:= (-1)_{F,B} \circ i_{W,O} \circ (-1)_{F,B}\circ i_{N,S} = i_{W,O} \circ i_{N,S} 
$$
 is indeed of order $3$.
To see this,  either compute the cube of its matrix, or use an argument following 
$((12)(23))^3 = ((12)(23)(12)) ((23)(12)(23))= (13)(13)=\id$. 
Concerning the antiholomorphic transformations, note that all matrices $M$ from the preceding tables are 
{\em unitary $2\times 2$ matrices  for the
Euclidean form on $\bA^2$}, i.e, they belong to the group 
\begin{equation}
\UU(2,\bA):= \bigl\{ M \in M(2,2;\bA) \mid \,  M^* M = 1 \bigr\} ,
\end{equation}
and hence commute with the orthcomplementation map $\zeta$ defined by this form. 
Since the real form $\tau^{NS}_O$ with respect to the Hermitian projective line is given by the skew-symmetric matrix $J$ (cf.\ Part I), it follows that
$\tau^{NS}_O= \zeta \circ J = \zeta \circ (-1)_{NS}$, that is, $\zeta = \tau^{NS}_O \circ (-1)_{NS}$.
The remaining formulae now follow from this.

 \end{document}